\documentclass{amsart}

\def\url#1{}
\def\urlprefix{}
\def\path#1{}
\def\newline{}

\usepackage{amssymb}

\usepackage[all]{xy}

\begin{document}
\title[Set systems : order types, continuous deformations, and quasi-orders]{Set systems: order types, continuous nondeterministic
deformations, and quasi-orders}

\author{Yohji Akama}
\address{
Mathematical Institute, Tohoku University
Sendai Miyagi JAPAN, 980-8578.\\
(tel) +81-22-795-7708\quad
(fax) +81-22-795-6400}
\email{akama@m.tohoku.ac.jp}

\date{\today}
\maketitle
\newcommand{\Nset}{\mathbb{N}}
\newcommand{\Qset}{\mathbb{N}}
\newcommand{\Zset}{\mathbb{Z}}

\def\Langle{\left\langle}
\def\R{\right\rangle}
\def\l{\langle}
\def\r{\rangle}
\def\pair#1#2{\l\,#1,\,#2\,\r}

\def\qo#1{\mathrm{qo}\left(#1\right)}
\def\Qo#1{\mathrm{Qo}\left(#1\right)}
\def\emptyseq{\langle\,\rangle}

\def\fld#1{\mathrm{fld}(#1)}

\newcommand{\A}{\ensuremath{\mathcal{A}}}
\newcommand{\B}{\ensuremath{\mathcal{B}}}
\newcommand{\QO}{\ensuremath{\mathbb{QO}}}
\newcommand{\SetSys}{\ensuremath{\mathbb{SS}}}
\newcommand{\RC}{\ensuremath{\mathbb{SS}}}
\newcommand{\coh}{\ensuremath{\mathbb{COH}}}

\newtheorem{question}{Question}
\newtheorem{definition}{Definition}
\newtheorem{fact}{Fact}
\newtheorem{example}{Example}
\newtheorem{theorem}{Theorem}
\newtheorem{corollary}{Corollary}
\newtheorem{lemma}{Lemma}
\newtheorem{remark}{Remark}
\newtheorem{proposition}{Proposition}
\def\transpose#1{{\,}^t{#1}}
\def\elang#1{\mathrm{eLang}_\Sigma({#1})}

\def\SGL{\mathcal{S}\textit{ingl}\,}
\def\DCL{\mathcal{D}\mathit{cl}\/}

\def\PosCon{{}}
\def\ewunion{\mathrel{\widetilde{\cup}}}
\def\ewuplus{\mathrel{\widetilde{\uplus}}}
\def\ewcap{\mathrel{\widetilde{\cap}}}
\def\ewbigcap{\mathrel{\widetilde{\bigcap}}}
\def\ewbigcupfin{\mathrel{\widetilde{\bigcup}_{finite}}}
\def\ewprod{\mathrel{\widetilde{\times}}}
\def\ewconcat{\mathrel{\widetilde{\cdot}}}

\def\ewstar#1{{#1}^{\widetilde{\ast}}\ }
\def\ewplus#1{{#1}^{\widetilde{+}}\ }

\def\ewshuffle#1#2{#1\mathrel{\widetilde{\diamond}}#2}
\def\ewshcls#1{{#1}^{\widetilde{\circledast}}\ }

\def\ewdiamond{\mathrel{\widetilde{\diamond}}}

\def\closure#1{\mathrm{ss}\left({#1}\right)}
\def\Ss{\mathrm{Ss}}
\def\ss#1{\Ss(#1)}
\def\upclos#1#2{#1\!\!\mathrel{\uparrow}\!\!#2}

\def\qed{}
\def\O{\mathcal{O}}

\newcommand{\op}{\ensuremath{\mathfrak{O}}}  
\newcommand{\discop}{\ensuremath{\mathfrak{F}}}
\newcommand{\TF}{\ensuremath{\mathcal{T}}}

\def\P#1{\{0,1\}^{#1}}

\def\C{\mathcal{C}}
\def\D{\mathcal{D}}

\def\natprod{\mathrel{\times\!\!\!\!\times}}

\def\E{\mathcal{E}}
\def\L{\mathcal{L}}
\def\M{\mathcal{M}}
\def\N{\mathcal{N}}

\def\S{\mathbf{S}}
\def\U{\mathbf{U}}

\def\X{\mathcal{X}}
\def\Y{\mathcal{Y}}
\def\Z{\mathcal{Z}}

\newcommand{\dom}{\mathrm{dom}}
\newcommand{\rng}{\mathrm{rng}}

\def\set#1#2{\bigl\{\;#1\;;\; #2\bigr\}}
\def\cardset#1#2{\#\bigl\{\;#1\;;\; #2\bigr\}}

\def\card#1{\#{#1}}

\def\otp{\mathrm{otp}}

\def\pit{positive information topology}

\def\inhfin#1{\left[{#1}\right]^{<\omega}}

\def\implies{{\ \ \mathrel{\Rightarrow}\ \ }}
\def\anonbin{{\mathrel{\odot}}}
\def\ewanonbin{{\mathrel{\widetilde{\odot}}}}
\def\aanonbin{\mathrel{{\anonbin}'}}
\def\eqref#1{(\ref{#1})}
\def\id{\mathrm{id}}

\def\Prd#1{\mathrm{Prod}(#1)}
\def\Prod#1{\mathrm{Prod}^*(#1)}
\def\Tch#1{\mathrm{Teach}(#1)}

\def\Ram{\mathrm{Ram}}

\def\inj{\mathsf{i}}

\begin{abstract} 
By reformulating a learning process of a set system $L$ as a game
 between Teacher and Learner, we define the order type of $L$
to be the order type of the game tree, if the tree is well-founded. The features of the order type of
 $L$ ($\dim L$ in symbol) are (1) We can represent any
 well-quasi-order~(\textsc{wqo} for short) by the set system $L$ of the
 upper-closed sets of the \textsc{wqo} such that the \emph{maximal order
 type} of the \textsc{wqo} is equal to $\dim L$.  (2) $\dim L$ is an
 upper bound of the mind-change complexity of $L$. 
 $\dim L$ is defined iff $L$ has a finite elasticity~(\textsc{fe} for
 short), where, according to computational learning theory, if an indexed family of recursive languages has
\textsc{fe} then it is learnable by an algorithm from positive data.
Regarding set systems as subspaces of Cantor spaces, we prove that
\textsc{fe} of set systems is preserved by any continuous function which
is monotone with respect to the set-inclusion. 
By it, we prove that finite
elasticity is preserved by various (nondeterministic) language
operators~(Kleene-closure, shuffle-closure, union, product,
intersection,$\ldots$.)  
The monotone continuous functions represent
nondeterministic computations. If a monotone continuous function
has a computation tree with each node followed by at most $n$ immediate
successors and the order type of a set system $L$ is $\alpha$, then the
 direct image of $L$ is a set system of order type at most $n$-adic diagonal Ramsey number of $\alpha$.
Furthermore, we provide an order-type-preserving
contravariant embedding from the category of quasi-orders and finitely
branching simulations between them, into the complete category of
subspaces of Cantor spaces and monotone continuous functions
having Girard's linearity between them.
\end{abstract}

\def\cmb#1#2{C_{#1}^{#2} }
\def\Rset{\mathbb{R}}
\def\F{\mathcal{F}}
\def\sh#1#2{#1\diamond #2}

\section{Introduction}
A set system $\L$ over a set $T$, a subfamily of the power set $P(T)$,
is a topic of (extremal)
combinatorics~\cite{MR866142,MR1931142}, as well as a target of an
algorithm to learn in computational learning theory of languages~\cite{Lange2008194}.

By reformulating a learning process of a set system $\L$ as a game
between $T$eacher and $\L$earner, we define
the order type of $\L\subseteq P(T)$ to be  the order type of the game tree. The features
of the order type of $\L$ ($\dim \L$ in symbol) are followings:
\begin{itemize}
\item We can represent any well-quasi-order~(\textsc{wqo} for short) by
      the set system of the upper-closed sets of the \textsc{wqo} such that the \emph{maximal order type}~\cite{MR0447056} of the \textsc{wqo}
      is equal to the $\dim \L$.

\item $\dim\L$ is an upper bound of the mind-change
      complexity~\cite{Ambainis1999323} of $\L$ which is recently studied in
relation to Noetherian property of algebras, set-theoretical topology and reverse
mathematics~\cite{luo06:_mind_chang_effic_learn,1715965,MR2640836,brecht09:_topol_and_algeb_aspec_of}. 
 $\dim L$ is defined if and only if $L$ has a finite elasticity~(\textsc{fe} for
 short), where, according to computational learning theory~\cite{114871,Lange2008194}, if an indexed family of recursive languages has
\textsc{fe} then it is learnable by an algorithm from positive data.
\end{itemize}

In computational learning of languages, a set system algorithmically
learnable from positive data is often a combination of set systems~(e.g. extended pattern languages~\cite{Lange2008194}.) 
To discuss which combinatorial operations for set systems preserve \textsc{fe}, quantitatively with the order type of the set systems, let us consider a motivating example. Suppose $\L$ is the class of 
arithmetical progressions over $\Nset$.  Observe the class of binary
unions of arithmetical progressions over $\Nset$, that is, $\L\ewunion\L:=\{
L\,\cup\, M\;;\; L,M\in \L\}$ is more difficult to learn than
$\L\;\widetilde \uplus\;\L:=\{\L\uplus M\;;\; L,M\in \L\}$, where
$L\uplus M$ is the disjoint union of $L$ and $M$, i.e., the union of the
progression $L$ colored red and the progression $M$ colored black. The
difficulty of $\L\ewunion\L$ is because the discoloration brings \emph{nondeterminism} to
Teacher and/or Learner.  By the discoloration of 
 $L\uplus M$, we mean $L\cup M$, and by that of
 $\L\;\widetilde{\uplus}\;\L$, we mean $\L\ewunion\L$. We can notice
 that the discolorization of the direct product $L\times M$ of languages
 $L,M$ is the concatenation $L\cdot M$, and observe that  $\L
 \ewprod \L= \{L\times M\;;\; L,M\in \L\}$ is easier to learn than  the
 discolorization $\L \ewconcat \L=\{ L\cdot M\;;\; L\in \L\}$. 

 Following questions are central in this paper:
\begin{question}\label{q1}
Does  discoloration preserve finite elasticity?
\end{question}

\begin{question}[\protect{\cite{Ng2008150,735102,Kanazawa_finite_elasticity94,kanazawa98:_learn_class_of_categ_gramm}}]\label{q2}
Which operations for set systems preserve finite elasticity?
\end{question} 

\begin{question}\label{q3}
What is the nondeterminism brought by operations that preserve finite
elasticity?
\end{question}
\begin{question}\label{q4}
How much do such operations increase the order type of set systems?
\end{question} 
 Question~\ref{q1} is yes, because Ramsey's theorem~\cite{Graham.Rothschild.ea:80} implies
  any dichromatic coloring of any infinite game sequence of $\L\ewunion\L$
 has an infinite, monochromatic game subsequence of $\L$.
This is another saying of
Motoki-Shinohara-Wright's theorem~\cite{93373,114871}. This argument leads to
  a solution of Question~\ref{q4} with Ramsey
  number~\cite{Graham.Rothschild.ea:80}. 

For Question~\ref{q2}, first observe that
the discoloration $L\cup M$ of $L\uplus M$ is the \emph{inverse image}
 $R^{-1}[L\uplus M]=\{ s\;;\; \exists u\in L\uplus M.\, R(s, u)\}$  by a following finitely branching 
 relation: $R(s,u):\iff u=\pair{s}{red}$ or $u=\pair{s}{black}$. For a
 relation $R\subseteq X\times Y$,
 the
 inverse images of a set $M$ and a set system $\M$ are, by definition, respectively
\begin{equation}
R^{-1}[M]:= \{x\in X\;;\; \exists y\in M.\ R(x,y)\},\ \
\widetilde{R^{-1}}[\M]:=\{ R^{-1}[M]\;;\; M\in\M\}. \label{def:invimage}
\end{equation}
 Let us abbreviate ``a set system with finite elasticity'' by
an \textsc{fess}.
In
\cite{Kanazawa_finite_elasticity94,kanazawa98:_learn_class_of_categ_gramm},
 Kanazawa derived ``the inverse image of an
\textsc{fess} by a finitely branching relation again an \textsc{fess}''
from K\"onig's lemma,
and established not only the union but also the permutation closure and
so on preserves \textsc{fess}s.
We generalize his lemma further as: ``the direct image $\L$ of an
\textsc{fess} $\M$ by a \emph{continuous} function which is monotone
with respect to the set-inclusion
 is again an
\textsc{fess}.''  Here we regard $\L$ and $\M$ as subspaces of Cantor spaces,
which are the product topological spaces $\P{\bigcup\L}$, $\P{\bigcup\M}$ of copies
of finite discrete topological space $\{0,1\}$.

Interestingly, a monotone, continuous function is a
stable function~\cite{Gir89} plus a modest
nondeterministic computation, so to say. To explain the relation among monotone,
continuous functions, (linear) stable functions and
nondeterminism, let us consider a following characterization by
Tychonoff's theorem:
a monotone, continuous function
 is a function $\op:\M\to\L$ such
that there
is a finitely branching relation $R\subseteq
(\bigcup\L)\times\inhfin{\bigcup\M}$ satisfying that for all $x\in \bigcup\L$
and all $M\in\M$,
\begin{displaymath}
 \op(1_M)(x)=\left\{\begin{array}{ll}1,\quad &\left( \exists v\subseteq M.\  R(x,v)
	      \right);\\
	                  0, &(\mbox{otherwise,})
\end{array}\right.
\end{displaymath} 
where $\inhfin{\bigcup\M}$ is the class of finite subsets of
 $\bigcup\M$ and $1_M$ is the indicator function of the set $M$.  From \emph{linear
logic}~\cite{Gir89} point of view, when $\L$ and $\M$ are coherence spaces and $\#\{v\;;\; R(x,v)\}\le 1$ for all $x$, then $\op$ becomes a stable function from $\L$ to
 $\M$, and if further $\forall x\forall v.\,
 (R(x,v)\Rightarrow \#v\le 1)$ holds, then $\op$ becomes a \emph{linear} stable
 function~\cite{Gir89}.  Kanazawa's lemma is nothing but ``the
direct image of an \textsc{fess} by a \emph{linear}, monotone, continuous function is again an \textsc{fess}'' where the relation
$R$ in the lemma is the \emph{trace}~\cite{Gir89} of the linear function.

For Question~\ref{q3}, the nondeterminism brought by the (linear)
monotone, continuous functions $\op$ are the ``finite
OR-parallelism'' caused by finite sets $v$'s. The degree of the
nondeterminism is $\#\{ v\;;\; R(x,v)\}$.  In other words, the trace $R$
of the monotone, continuous function is finitely branching, while
that of stable function has at most one branching. So we can easily
prove that there are
monotone, continuous functions $1_L\mapsto 1_{L^*}$  and
$1_L\mapsto 1_{L^\circledast}$ where $L^\circledast$ is the
shuffle-closure~\cite{MR2078698} of $L$.
 Here
are a non-example and an example of nondeterminism.
\begin{itemize}
\item
Because a $\Pi$-continuous function~\cite{luo06:_mind_chang_effic_learn}
     can represent an unbounded search unlike monotone, continuous functions, the direct image of an
     \textsc{fess} by a $\Pi$-continuous function
 is not necessarily an
\textsc{fess}~(see Theorem~\ref{thm:7}.)

\item
We define the category  $\QO_{FinSim}$
 of \emph{quasi-orders} and finitely branching simulations between
them. Here a
usual order-homomorphism is an instance of a finitely branching
     \emph{simulation} which appears in concurrency theory.
Let $\SetSys$ be
 the complete category of set systems
and monotone, continuous functions between them.
 We provide  an
\emph{order-type-preserving} contravariant embedding from $\QO_{FinSim}$ to
$\SetSys$. By this embedding, each quasi-order is sent
to the family of upper-closed sets.
When the  branching of the relation is at most 1, it is sent to
 a stable~(\emph{sequential}) function~\cite{Gir89} in
 $\SetSys$. In fact, the category of coherence spaces and stable
 functions between them, introduced in \cite{Gir89} embeds in $\SetSys$.
\end{itemize}



As for Question~\ref{q4}, the Ramsey number argument for
Question~\ref{q1} establishes : If a monotone, continuous
function $\op$ with the trace $R$ has $n$  such $\#\{ v\;;\; R(s,v)\}\le
n$ for
each $s$, then the
direct image of $\L$ by $\op$ has order type at most the $n$-adic diagonal Ramsey
number of $\dim \L+2$.

This paper is organized as follows. In the next section, we review parts
of order theory, various (closure) operations of languages from
algebraic theory~\cite{MR2078698,MR2567276} of languages and automata, and finite elasticity of
computational learning theory.  In Section~\ref{sec:intersection}, we
introduce the order type of a set system, and then represent every quasi-order by a set system
having the same order type as the quasi-order.  We prove that if the set system $\L$ is an indexed
family of recursive languages, as in the case of computational learning
theory,  and if moreover the indexing is without
repetition, then $\dim\L$ is exactly a recursive ordinal. 
In Section~\ref{sec:cont_image}, we prove ``the direct image of an
\textsc{fess} by a monotone, continuous function is again an
\textsc{fess}.''  
In Section~\ref{sec:closureop}, we employ Ramsey numbers to answer
Question~\ref{q4}.
In Section~\ref{sec:embed}, we embed the category
$\QO_{FinSim}$ and a categorical model of linear logic in the category
$\SetSys$. 
In \ref{sec:cat},
we record the proof of Theorem~\ref{thm:catmain} on the categorical structure of
$\SetSys$, $\SetSys_{lin}$ and $\SetSys_{seq}$, where $\SetSys_{lin}$ is
the subcategory induced by linear functions and $\SetSys_{seq}$ by
sequential functions.  We prove the category $\SetSys_{seq}$ does not have a
binary coproduct because the sequential function does not represent
a nondeterministic computation. And then we discuss whether $\SetSys$ has
the duality operator and the bang operator as the category of coherence
spaces. 

\section{Preliminaries}\label{sec:preliminaries}

Let $R\subseteq \S\times\U$ be a relation. If the cardinality $B_R(s)$ of $\{ u\in
 \U\;;\; R(s,u) \}$ is finite for all $s\in \S$, then we say $R$ is
 \emph{finitely branching}. If 
 $B_R(s)\le 1$ for all $s\in \S$, then we say $R$ is a \emph{partial
 function}.
For a set $\U$, let 
$\left[\U\right]^{<\alpha}$ be the class of subsets $A$ of $\U$ such
 that $\#A <\alpha$.

\subsection{Order theory}
A quasi-order~(\textsc{qo} for short) over a set $X$ is a pair $\X=(X, \preceq)$ where
$\preceq$  is a reflexive, transitive  relation. A \emph{bad} sequence is
a possibly infinite sequence $\langle a_0, a_1, \ldots , a_n(,
\ldots)\rangle$ such that $a_i \not \preceq a_j$ whenever $i < j$. A
\emph{well-quasi-order} (\textsc{wqo} for short) is a quasi-order that
has no infinite bad sequences. For  $A\subseteq X$, let $A\uparrow\X:=\{ x\in \X\;;\;\exists a
\in A.\, a\preceq x\ \}$.

\begin{definition}For a quasi-order $\X=(X,
\preceq)$, let a set system $\closure{\X}$ be the complete lattice of
 upper-closed subset of $X$ with respect to $\X$.
\end{definition} 

\begin{proposition}[\protect{\cite[Theorem~2.1]{MR0049867}}]\label{prop:higman}
 For every quasi-ordered set $\X$, the following are equivalent:
\begin{enumerate}
\item
 $\X$ is a \textsc{wqo}. 
\item \emph{Finite basis property}:
Every $A \uparrow \X$ is $B\uparrow \X$ for some $B\in\inhfin{X}$. 

\item \emph{Ascending chain condition}:
$\closure{\X}$ is a complete lattice with ascending chain condition. That is,
there is no infinite, strictly ascending sequence of members.
\end{enumerate} 
\end{proposition}

 The length of a sequence $\sigma= \l b_1,
\ldots, b_m\r$ is, by definition, $ln(\sigma)=m$, and the length of an
infinite sequence $\sigma$ is, by definition, $ln(\sigma)=\infty$.

By a \emph{tree}, we mean a set $T$ of finite sequences such that any
initial segment of a sequence in $T$ is  in $T$. A tree $T$ is said
to be \emph{well-founded} if there is no infinite sequence $\langle
a_1,a_2,\ldots\rangle$  such that $\langle a_1,\ldots,a_n\rangle$ is in
$T$ for each $n$.

Let $T$ be a well-founded tree. For each node $\sigma$ of $T$, let the
ordinal number $| \sigma |$ be the supremum of $|\sigma'| + 1$ such that
$\sigma'\in T$ is an immediate extension of $\sigma$. Then the
\emph{order type} $|T|$ of the well-founded tree $T$ is defined by the
ordinal number $|\emptyseq|$ assigned to the root $\emptyseq$ of $T$.
For a tree $T$ which is not well-founded, let $|T|$ be $\infty$. For the sake of convenience, we set $\alpha < \infty$ for all ordinal
numbers $\alpha$. As in \cite{MR1283862}, we define
the \emph{order type} $\otp(\X)$ of a \textsc{wqo} $\X$ to be the order
type of the well-founded tree of bad sequences in $\X$. According to
\cite[Sect.~2]{MR1283862}, $\otp(\X)$ is equal to
the \emph{maximal order type} of de Jongh-Parikh~\cite{MR0447056}.

By an \emph{embedding} from a tree $T$ to a tree $T'$, we mean an injection
$f:T\to T'$ such that $f(v\sqcup u)=f(v)\sqcup f(u)$ for all vertices
$u,v$ in $T$, where $v\sqcup u$ is the greatest common ancestor of a
pair of vertices $u,v$.
\begin{fact}\label{fact:TE}If there is an embedding from a tree $T$ to
 a tree $T'$, then $|T|\le |T'|$.\end{fact} 

\subsection{Computational learning
theory for languages} 

A \emph{set system} is a subfamily of a power set. We use
$\L,\M,\N,\ldots$ to represent set systems.

We say
a set system $\L$ over $X$ has an \emph{infinite elasticity}, if
there are infinite sequences $t_0, t_1,\ldots\in X$ and
$L_1,L_2,\ldots\in \L$ such that $\{t_0,\ldots, t_{i-1}\}\subseteq
L_i\not\ni t_i$ for every positive integer $i$. Otherwise, we say
$\L$ has a \emph{finite elasticity}~(\textsc{fe}.) A set system with an
\textsc{fe} is abbreviated as an \textsc{fess}.

Let $\Nset$ be the set of nonnegative integers.
\begin{example}\label{eg}\rm
\begin{enumerate}

\item The class of integer lattices contained in $\Zset^d$ and the class of
      ideals over $\Zset[x,y,z]$ are 
      \textsc{fess}s, because $\Zset^d$ and $\Zset[x,y,z]$ are is a
      Noetherian module and a Noetherian ring respectively~\cite[p.~112]{MR1790326}. 

\item The class of finitely generated free sub-semigroups of $(\Nset^2,
      +)$ is not
      an \textsc{fess}~(\cite{DBLP:conf/lata/Akama09}.)

\item The class of (extended) pattern languages with bounded number of
      variables is an indexed family of recursive languages, and is an
      \textsc{fess}~(\cite{114871}. For an elementary proof, see \cite{DBLP:conf/lata/Akama09}.)

\item $\SGL:=\left\{ \{x\}\;;\; x\in \Nset\right\}$ of
 singletons is an \textsc{fess}.

\item The class
$\DCL:=\left\{ \{y\;;y\le x\}\;;\; x\in \Nset\right\}\subseteq
 P(\Nset)$ is not an \textsc{fess}.
\end{enumerate} 
\end{example} 


\begin{definition}\label{def:aho}
By an \emph{indexed family of recursive languages~(\textsc{ifrl}
for short)}, we  mean
a pair $\L=(\nu:J\to X,\ \gamma:I\times J\to \{0,1\})$ such that
 $I,J\subseteq\Nset$, $\nu$ is a bijection, and $\gamma$ is recursive.
Put  $\L^i:=\{\nu(j)\in X\;;\; \gamma(i,j)=1\}$ for
       $i\in I$. An \textsc{ifrl} without repetition is just an
 \textsc{ifrl} such that $\L^i\ne L^j$ for distinct $i,j$.
\end{definition}

\begin{proposition}[\protect{\cite{93373,114871}}]  Every \textsc{ifrl} with an \textsc{fe} is learnable from
 positive data by an algorithm.\end{proposition}


By an alphabet we mean a finite nonempty set.
Let $\Sigma$ be an alphabet. Denote the empty word by $\varepsilon$.
For words $u,v\in \Sigma^*$, the
\emph{shuffle product} $u \diamond v$ of $u$ and $v$ is, by definition,
the set of all the words $u_1 v_1 u_2 v_2...u_n v_n$ such that $\exists n \geq 1 \exists
u_1,u_2,\ldots,u_n,v_1,v_2,\ldots,v_n \in \Sigma^*$ we have $u = u_1 u_2 \cdots u_n$ and  $v = v_1
v_2 \cdots v_n$.
 For $L,M\subseteq \Sigma^*$, let $L\diamond M:= \bigcup \{ u\diamond v\;;\;
u\in L, v\in M\}$. Put
$L^\diamond:= L\cup (L\diamond L)\cup (L\diamond L\diamond
L)\cup\cdots\enspace$. Let us call $L^\circledast := L^\diamond\cup\{\varepsilon\}$ \emph{the shuffle-closure of $L$}. The
shuffle-product and shuffle-closure are studied in 
algebraic theory of automata and languages~\cite{MR2078698,MR2567276}, for example.

Let a disjoint union of languages $B_i$ $(i\in I)$ be
\begin{displaymath}
 \biguplus_{i\in I} B_i:=\{ \pair{b}{i}\;;\; b\in B_i, i\in I\ \}.
\end{displaymath}

For a language $M$, let $M^m$ be $\overbrace{M\cdot M\cdot\cdots \cdot
M}^m$ $(m\ge1)$, let us call $M^+ := \bigcup_{m\ge1}M^m$ \emph{the positive Kleene-closure},
and let $M^*$ be the Kleene-closure. Let $M^\circledast$ be
the shuffle-closure of $M$, 
$\frac{1}{2}(M)$ be the half initial segment.

 For all $\L_i\subseteq P(X_i)$  $(1\le i\le n)$ and an operation
 $\anonbin$ on languages of arity $n$, put
\begin{displaymath}
\mathrel{\ewanonbin}(\L_1,\ldots, \L_n):=\{\anonbin(L_1,\ldots,L_n)\;;\;
 L_i\in \L_i, (1\le i\le n)\ \}.
\end{displaymath} 

%
Here is an application of Ramsey's theorem:

\begin{proposition}[\protect{Motoki-Shinohara-Wright~\cite{93373,114871}}]\label{prop:msw}
If $\L_1$ and $\L_2$ are \textsc{fess}s, so is
$\L_1\ewunion\L_2$.
\end{proposition}

In fact, it is derived from a weak principle: K\"onig's lemma.

\begin{proposition}[\protect{Moriyama-Sato~\cite{735102}}] \label{prop:a}
For a fixed finite alphabet, the family of language classes with
\textsc{fe} is closed under $\cup$, $\ewunion$, $\widetilde{\cdot}$, $\ewcap$,
${\,}^{\widetilde*}$, ${\,}^{\widetilde+}$, and  ${\,}^{\widetilde m}$
 for every positive integer $m$, but not under
elementwise complement.
\end{proposition} 

Independently, 
Kanazawa~\cite{Kanazawa_finite_elasticity94,kanazawa98:_learn_class_of_categ_gramm}
proved  a following nice result by using K\"onig's lemma: 

\begin{proposition}[\protect{Kanazawa~\cite{Kanazawa_finite_elasticity94,kanazawa98:_learn_class_of_categ_gramm}}]\label{prop:5}
If $\M\subseteq P(Y)$ is an \textsc{fess} and $R \subseteq X \times Y$
is finitely branching, then $\widetilde{R^{-1}}[\M]\subseteq P(X)$ is so.
\end{proposition}

In fact, without invoking K\"onig's lemma, he showed
\begin{lemma}[\protect{\cite{Kanazawa_finite_elasticity94,kanazawa98:_learn_class_of_categ_gramm}}]\label{lem:a} If $\dim \L, \dim \M<\infty$, then
  $\dim \L\ewuplus \M$. 
\end{lemma}
Then he proved various language operations preserves \textsc{fess}s, by 
applying Proposition~\ref{prop:5}.

\begin{corollary}[\protect{Kanazawa~\cite{Kanazawa_finite_elasticity94,kanazawa98:_learn_class_of_categ_gramm}}] \label{cor:kanazawa}For a fixed alphabet,
the family of language classes with an \textsc{fe} is closed under
 $\ewunion$, elementwise permutation closures,
 $\ewdiamond$, and $\widetilde{\frac{1}{2}}(\cdot)$. For each nonerasing homomorphism
 $h:\Sigma_1^*\to\Sigma_2^*$, if a language class $\L \subseteq
 P(\Sigma_1^*)$ has an \textsc{fe}, so does $\widetilde{h}[\L]\subseteq
 P(\Sigma_2^*)$. If $\L$ is a class of $\varepsilon$-free languages with
 an \textsc{fe}, then so is $\{ L_1\cdot L_2 \cdot \cdots\cdot L_n\;;\; n\ge1, L_1,\ldots,L_n\in\L \}$.
\end{corollary}

\section{Order types of set systems and WQOs} \label{sec:intersection}

We introduce order types of set systems, study
the set system of upper-closed
subsets of a \textsc{qo} from viewpoint of order types and algebraic
theory of lattices~\cite{0494.06001}.

We regard a learning process of a set system $\L\subseteq P(T)$, as a game
between Teacher $T$ and Learner $\L$ where in each inning $i\ge1$
Teacher presents a ``fresh'' example $e_{i-1}\in T$ and Learner submits
a hypothesis $H_i\in \L$ that explains examples presented so far, that
is, $\{e_0,\ldots, e_{i-1}\}\subseteq H_i$. By a ``fresh''example
$e_{i-1}$, we mean $e_i\not\in H_i$.  The well-foundedness of
the game tree coincides with the \emph{finite elasticity}~\cite{93373,114871} of the set system
$\L$, which was introduced in computational learning theory of
languages~\cite{Lange2008194}. 
 If $\L$ is further an \textsc{ifrl}, 
then some algorithm can learn $\L$ from positive data~\cite{93373,114871}.  First, we
introduce \emph{the order type $\dim \L$ of the set system} $\L$ by the
order type of the game tree.

\begin{definition}[Production sequence]A \emph{production sequence} of a
 set system $\L$ is  a sequence  $\Langle
 \pair{t_0}{L_1}, \pair{t_1}{L_2},  \ldots, \pair{t_{m-1}}{L_m}\R$ $(m\ge 0)$ or an infinite
 sequence 
  $\Langle \pair{t_0}{L_1},
\l t_1 , L_2\r, \ldots \R$ such that  

\begin{displaymath}  \{t_0,
 \ldots , t_{i-1}\} \subseteq L_i \in \L\ (i=1,2,\ldots(,m))\quad \mbox{and}\quad
 L_j\not\ni t_j\ (j=1,2,\ldots(,m-1).)\end{displaymath}
Let $\Prd{\L}$ be the set of all production sequences of $\L$.
\end{definition} 

Clearly a sequence $\langle L_1,L_2, \ldots\rangle$ is a bad sequence in
a poset $(\L, \supseteq)$, because $i<j$ implies $L_j\setminus L_i\ni t_i$.

\begin{definition}[Dimension]
The \emph{dimension} of $\L$, denoted by $\dim \L$, is defined to be $|\Prd{\L}|$.
\end{definition} 

By Fact~\ref{fact:TE}, $\L \subseteq \L'$ implies  $\dim \L\le \dim
\L'$.

Let us see examples of order types of set systems.

We note that any ordinal $\alpha$ is the dimension of the set
system of upper-closed subsets of $\alpha$. 

As in \cite[p.~384]{MR982269}, we understand that a \emph{recursive
ordinal} is an ordinal number $\alpha$ such that $\alpha=|T|$ for some
recursive well-founded tree $T$. 

\begin{theorem}\label{thm:omega1ck}If an \textsc{ifrl} $\L$ without repetition is an \textsc{fess}, then $\dim\L$ is
 a recursive ordinal. Conversely, for every recursive ordinal $\alpha$
 there is an \textsc{ifrl} $\L$ without repetition such that $\alpha=\dim
 \L $.
\end{theorem}
\begin{proof}
Let $\L$ be an \textsc{ifrl} without repetition by a pair of functions
 $\nu, \gamma$. 
Define a set $T_\L\subset \Nset$ inductively as follows. We also use
 symbols  '$\l$' and '$\r$' for sequence numbers, and 
 Odifreddi's notation~\cite[p.~88]{MR982269} of operations
 on sequence numbers.
(1)
 $\emptyseq\in T_\L$. (2) If  $\gamma(e,j_0)=1$, then
 $\left\langle \langle j_0, e \rangle\right\rangle \in T_\L$.
(3) If $\sigma\in T_\L$, $\gamma\left(\left((\sigma)_{ln
 (\sigma)-1}\right)_1, j\right)=0$, and 
 $\gamma(e,j)=1=\gamma\left(e,\ \left( (\sigma)_k\right)_0 \right)$
 $\forall k<\ln(\sigma)$, then $\sigma\ast \left\langle\langle
 j,e\rangle\right\rangle\in T_\L$. Clearly $T_\L$ is a recursive tree. 

Let $\varphi:\Prd{\L}\to T_\L$ take any $\Langle \l t_0,
 L_1\r,\l t_1, L_2 \r,\ldots,\l t_{l-1}, L_l\r\R\in \Prd{\L}$ to a node $\left\langle
 \pair{j_0}{e_1}, \pair{j_1}{e_2}, \ldots, \pair{j_{l-1}}{e_l}\R$ of $T_\L$ where $e_i$ is the unique
 number such that $L_i=\L^{e_i}$ and $j_i$ is
 the unique number such that $\nu(j_i)=t_i$. Then $j_i$ is well-defined
 because $\nu$ is bijective, and  $e_i$ is too because $\L$ is an \textsc{ifrl} \emph{without
 repetition}. The function $\varphi$ is obviously an surjective
 order-homomorphism that preserves glb's, and in fact an injection because $\L$ is an
 \textsc{ifrl} without repetition. Therefore $\dim \L= |T_\L|$.  Since
 $\L$ is an \textsc{fess}, $T_\L$ is a recursive well-founded tree, so
 $\dim \L$ is a recursive ordinal number.

Next we prove the second assertion. The ordinal number $\alpha$ is
 constructive by \cite[Theorem XX, Ch.11]{MR0224462}. So there is a
 recursively related, univalent system assigning a notation to $\alpha$
 by \cite[Theorem XIX, Ch.11]{MR0224462}. Therefore, there is
 an injective function $\nu$ from some set $J\subseteq\Nset$ onto
 an initial segment $\{\beta\;;\;0\le\beta\le\alpha\}=\alpha+1$ of the
 ordinal numbers such that 
\begin{equation}
\{\langle x,y\rangle\;;\; x,y\in J,\ 
 \nu(x)\le\nu(y)\}\subseteq\Nset\ \mbox{is recursive}. \label{recrela}
\end{equation} 
In particular, $J$ is recursive.  When $J$ is infinite, there is a
recursive strictly monotone function $d$ with the range being $J$. Put
$\L :=\{ L^i\;;\; i\in\Nset\}$, and $L^i:=\{ \beta\;;\; \exists j\in J.\
\beta=\nu(j)\ge \nu(d(i))\ \}$. Because $\nu$ is a bijection to
$\alpha+1$, $L^i$ is an upper-closed subset of $\alpha+1$. Define a
function $\gamma:\Nset\times J\to \{0,1\}$ by $\gamma(i,j)=1$ if
$\nu(j)\ge \nu(d(i))$, 0 otherwise. From \eqref{recrela}, $\gamma$ is
recursive.  In fact, $\L$ is $\closure{(\alpha+1, \le)}$ because the
range of $d$ is exactly $J$ and $\nu(J)=\alpha+1$. Then $\L$ is an
\textsc{ifrl}. Moreover $\L$ is an \textsc{ifrl} without repetition
because $d$ and $\nu$ are injective. By
Theorem~\ref{thm:repre}~\eqref{assert:otp}, $\dim\L=\otp\left((\alpha+1,
\le)\right)=\alpha$.
When $J$ is finite, we can prove the assertion similarly.
\qed
\end{proof}

Next we introduce a left-inverse of $\closure{\bullet}$.

\begin{definition}\label{def:ss_to_qo}
For a set system $\L \subseteq P(X)$, define a quasi-order
\begin{displaymath}
x\preceq_\L y:\iff \forall L\in \L\, (x\in L\Rightarrow y\in L.)\quad
 \qo{\L}:=(X, \preceq_\L.)
\end{displaymath}
\end{definition} 

Below, we prove that
$\closure{\bullet}$ is an order-type preserving representation of
\textsc{qo}s by set systems. In other words, the order type of a \textsc{wqo} turns out to be the difficulty in learning the
class of upper-closed subsets of the \textsc{wqo}. Then we prove that $\closure{\bullet}$ indeed has $\qo{\bullet}$ as the left-inverse.
\begin{theorem}[Representation of QO]\label{thm:repre}
Let $\X=(X,\preceq)$ be a quasi-order. 
\begin{enumerate}

\item \label{assert:otp}
     $\otp(\X) = \dim \closure{\X}$.

\item \label{assert:inj}  $\X= \qo{\closure{\X}}$. 
\end{enumerate}
\end{theorem} 

\begin{proof}

We prove the assertion~\eqref{assert:otp},
by a transfinite induction using
\begin{eqnarray*}
\exists L_1,\ldots, L_l.\ 
\Langle \l t_0, L_1\r, \l t_1, L_2 \r, \ldots,\l t_{l-1}, L_l \r\R \in
 \Prd{\closure{\X}}\\
\iff\l t_0,\ldots,t_{l-1}\r\mbox{ is a bad sequence of $\X$}. 
\end{eqnarray*}
The $\Rightarrow$-part is demonstrated as follows:
 We have $\{t_0,\ldots,
 t_{i-1}\}\subseteq L_i\not\ni t_i$ $(1\le i\le l-1)$. Because each
 $L_i\in \closure{\X}$ is
 upper-closed, for any nonnegative integers $j<i\le l-1$, $t_j\not\preceq
 t_i$.  The $\Leftarrow$-part is witnessed by $L_i:=\upclos{\{t_0,\ldots,t_{i-1}\}}{\X}$.

\eqref{assert:inj}
 Assume $x\preceq_{\closure{\X}} y$. Then $\forall L\in
 \closure{\X}.\ (x\in L \Rightarrow y\in L)$. Take $L:=\{ y\in
 X\;;\; x\preceq y\}\in \closure{\X}$. Hence $x\preceq y$. 
Conversely, assume $x\preceq y$. Then because every $L\in \closure{\X}$
 is upper-closed with respect to $\preceq$,   $x\in L$ implies $y\in
 L$. Therefore $x\preceq_{\closure{\X}} y$. \qed
\end{proof}

\begin{theorem}\label{rem:false}
 If $\qo{\L}$ is a
 \textsc{wqo}, $\L$ is an \textsc{fess} but not conversely. Actually
$\dim\L\le \otp(\qo{\L})\le \infty$ and $\dim (\SGL)=1<
 \otp(\qo{\SGL})$.
\end{theorem} 
\begin{proof}
Observe that for every
$\Langle \l t_0, L_1\r, \l t_1, L_2 \r, \ldots,\l t_{l-1}, L_l \r\R\in
\Prd{\L}$, a sequence $\l t_0,
 t_1,\ldots, t_{l-1}\r$ is a bad sequence of $\qo{\L}$. So we can prove
 the inequality
by a transfinite induction~\cite{MR0219424} on $\Prd{\L}$.  The equality $\dim\L= \otp(\qo{\L})$ is not necessarily
 true. For example, although $\dim \SGL=1$, a quasi-order
 $\qo{\SGL}=(\Nset, =)$ has an infinite bad sequence $\l 0,1,2,3,\ldots\r$,
 which implies $\otp (\qo{\SGL})=\infty$. 
\end{proof}

\begin{proposition}[\protect{\cite[p.~41]{brecht09:_topol_and_algeb_aspec_of}}]
If $\L$ has a finite thickness and $\L$ has no infinite anti-chain with
 respect to $\subseteq$, then $\qo{\L}$ is a \textsc{wqo}.
\end{proposition}

We will study structure of 
the representation of \textsc{qo}s by set systems from viewpoint of algebraic theory of
lattices~\cite{0494.06001}. From \cite{0494.06001}, we recall ``atom,''
``atomic,'' and ``compact'' (and the dual notions.)

Let $L$ be a complete lattice.
By a \emph{coatom} of  $L$, we mean any
nontop element $C$ such that every nontop $c\in
L$ is codisjoint from $C$ (i.e. $c\cup C$ is top) or less than or equal to
$C$. A \emph{coatomic}, complete lattice is, by definition, a complete lattice such that
for any nontop element $C_0$ there is a coatom greater than or equal to
$C_0$. We say an element $c$ in a complete lattice $L$ is called
\emph{compact} if whenever $c\le \bigcup S$ there exists a
      finite subset $T\subseteq S$ with $c\le \bigcup
      T$. 
\begin{proposition}[\protect{\cite{0494.06001}}]\label{prop:Dilworth}
Every element of a complete lattice $L$ is compact if and only if $L$
satisfies the ascending chain condition.
\end{proposition}

\begin{theorem}\label{thm:repr}
Let $\X=(X,\preceq)$ be a quasi-order. 
\begin{enumerate}
\item \label{assert:tfae}The following are equivalent:
\begin{enumerate}
\item 
$\X$ is a
 \textsc{wqo}.
\item  $\closure{\X}$ is an \textsc{fess}. 

\item $\closure{\X}$ is a complete lattice such that every element is
      compact.
\end{enumerate}

\item \label{assert:coatomic}If $\X$ is a
 \textsc{wqo}, then
$\closure{\X}$
 is a coatomic, complete lattice.

\end{enumerate}
\end{theorem} 
\begin{proof}
As for the assertion~\eqref{assert:tfae}, the equivalence between the conditions (a) and (b)
 follows from Theorem~\ref{thm:repr}~\eqref{assert:otp}. The equivalence between
 the conditions (a) and (c) is
  by Proposition~\ref{prop:higman} and Proposition~\ref{prop:Dilworth}.

\eqref{assert:coatomic} Let $\X=(X,\preceq)$. $\{\emptyset, X\}$ is obviously a coatomic,
 complete lattice. So assume $\closure{\X}\ne \{\emptyset, X\}$. By
 the assertion~\eqref{assert:otp},
 the complete lattice
 $\closure{\X}$ is an \textsc{fess}. 
If $\closure{\X}$ is not coatomic,
 then there exists $C_0 \in \closure{\X}\setminus\{X\}$ such that
\begin{equation}
\forall C\in \closure{\X}\setminus\{X\}\ \bigl(C_0\subseteq
 C\ \Longrightarrow\ \exists c \in\closure{\X}\ \left( c \cup C\ne
 X \ \&\  c \setminus
 C\ne\emptyset\right)\bigr).\label{raa:coatom}
\end{equation}
We can construct an infinite $\Langle \l x_0,C_1\r,\l x_1, C_2\r, \ldots\R\in \Prd{\closure{\X}}$
 as follows: Because $\closure{\X}\supsetneq\{\emptyset, X\}$,  we can
 take a pair of
 $C_1\in \closure{\X}\setminus\{X\}$ and $x_0\in C_1$
 such that $C_0\subseteq C_1$. Suppose we have a pair of $C_i\in
 \closure{\X}\setminus\{X\}$ and $x_{i-1}\in C_i$ such that
 $C_0\subseteq C_i$. Once we can find a pair of $C_{i+1}\in
 \closure{\X}\setminus\{X\}$ and $x_i\in C_{i+1}\setminus C_i$ such that
 $C_0\subseteq C_{i+1}$, then by
 iterating this process, we can construct an infinite production
 sequence of ${\closure{\X}}$.
 Because $C_0\subseteq C_i$ and 
 \eqref{raa:coatom}, there exist $c_i\in \closure{\X}$ and $x_i\in
 c_i\setminus C_i$ such that $c_i\cup C_i\ne X$. So, let
 $C_{i+1}:=C_i\cup c_i$. Then it is in $\closure{\X}\setminus\{X\}$ because
 $\closure{\X}$ is closed under the union. Moreover $x_i\in
 C_{i+1}\setminus C_i$ because $x_i\in c_i\setminus C_i$. Clearly
 $C_0\;\subseteq\; C_i\;\subseteq\; C_i\cup c_i=C_{i+1}$.
\qed
\end{proof}

This section suggests a close similarity between \textsc{wqo}s and 
finitely elastic set systems, so it is worth studying
 whether the
closure properties for \textsc{wqo}s solve the questions of which operation on set systems 
preserves finite elasticity. According to \cite{MR1189832},
the study on
closure properties for \textsc{wqo}s~(Higman's theorem for \textsc{wqo}s on finite sequences~\cite{MR0049867},
 Kruskal's theorem for \textsc{wqo}s on finite trees~\cite{MR0306057}, Nash-Williams' theorem for \emph{better-quasi-order}s
on transfinite sequences~\cite{MR0221949},...) can be advanced via set-theoretic topological methods and a Ramsey-type
argument. 
So, to advance the study on  the questions of which operation on set system 
preserves finite elasticity, 
it is natural for us to employ set-theoretic topology (see Section~\ref{sec:cont_image}) 
and a Ramsey-type argument~(see Section~\ref{sec:closureop}.)

\section{Continuous deformations of set systems}
\label{sec:cont_image}

For nonempty finite set $U$, the product topological space $\{0,1\}^U$
is called a \emph{Cantor} space. 
Subspaces of Cantor spaces are
represented by $\C, \D, \E,\ldots$.

\begin{definition}\label{def:inj} 
For every set system $\L\subseteq P(X)$, define a function 
\begin{displaymath}
 \inj\;:\; \L\to 
\inj \L:= \left\{1_L\in \P{\bigcup\L}\;;\; L\in
 \L\right\}\ ;\ L\mapsto 1_L\enspace.
\end{displaymath}
Then $\inj \L$ is a topological space, induced from a Cantor space
$\P{\bigcup \L}$.  For  $\C\subseteq \P{X}$, put 
\begin{displaymath}
\fld{\C}:=\bigcup \inj^{-1}(\C)\subseteq X. 
\end{displaymath}
\end{definition} 

Let us identify $g\in \D$ with an infinite sequence
 $\left(g(y)\right)_{y\in\fld{\D}}$. For each $x\in \fld{\C}$, let
 $\pi_x:\C\to\{0,1\}$ be the canonical projection to the $x$-th
 component. So $\pi_x(f)=f(x)$ for every $x\in\fld{\C}$. Recall that a
 Cantor space $\P{\fld{\C}}$ is generated by a class of sets
 $\pi_x^{-1}[\{b\}]$ such that $x\in\fld{\C}$ and $b\in\{0,1\}$. Let us
 call each $\pi_x^{-1}\left[ \{b\}\right]$ a generator of $\C$. Then an
 open set of $\C$ is exactly an arbitrary union of finite intersections
 of generators. Note that each generator of $\C$ is clopen.
 
A \emph{Boolean formula over a set $Y$} is built up from the truth values 0, 1, or elements of
$Y$,  by means of negation, finite conjunction, and finite disjunction.

\begin{lemma}\label{prop:b}
A function
$\op:\D\to\C$ is continuous, if and only if
 there is a sequence $(B_x)_{x\in \fld{\C}}$ of Boolean formulas over
 $\fld{\D}$  such that for every $g\in\D$ and every $x\in\fld{\C}$,  the value
 $\op(g)(x)$ is the truth value of $B_x$ under the truth assignment
 $g$. 
\end{lemma}
\begin{proof}
(If-part) The inverse
 image $\op^{-1}\left[ \pi_x^{-1}\left[\{b\}\right]\right]$ of a generator $\pi^{-1}_x\left[\{b\}\right]$ is the class of the
 truth assignments $g\in \D$ under which  the truth value of $B_x$  is $b$. Because the Boolean formula $B_x$ is equivalent to a finite
 disjunction of finite conjunctions of elements of $\fld{\D}$ and the
 negations of elements of $\fld{\D}$, the inverse image $\op^{-1}\left[
 \pi_x^{-1}\left[ \{b\}\right]\right]$ is just a finite union of finite
 intersections of generators of $\D$, while $\op^{-1}\left[
 \pi_x^{-1}\left[ \{0\}\right]\right]$ is just a finite intersection of
 finite unions of generators of $\D$. Therefore, the inverse image
 $\op^{-1}\left[ \pi_x^{-1}\left[\{b\}\right]\right]$ is open. 

(Only-if-part) Because $\op$ is continuous and $\{b\}$ $(b=0,1)$ is clopen in the
 finite discrete topology $\{0,1\}$, the inverse image $\op^{-1}\left[ \pi_x^{-1}\left[\{b\}\right]\right]$ of a
 generator $\pi_x^{-1}\left[\{1\}\right]$ by $\op$ is clopen, which is an
 arbitrary union of intersections of generators. 

Because $\{0,1\}$ is compact, Tychonoff's theorem implies the compactness
 of 
 $\{0,1\}^{\fld{\C}}$ and thus that of $\C$. Moreover, $\C$ is a
 Hausdorff space, because for all distinct $f,g\in\C$, there is
 $x\in\fld{\C}$ such that $f(x)\ne g(x)$, which implies that $
 \pi_x^{-1}\left[ \{ f(x)\}\right]$ and $ \pi_x^{-1}\left[ \{
 g(x)\}\right]$ are open sets such  that $f\in
 \pi_x^{-1}\left[ \{ f(x)\}\right]$ and $g\in \pi_x^{-1}\left[
 \{g(x)\}\right]$.

Since every closed subset of a compact Hausdorff space is compact, the
 clopen set $\op^{-1}\left[\pi_x^{-1}\{b\}\right]$  is
  $\bigcup_{i=1}^m\bigcap_{j=1}^{n_i} \pi_{y_{ij}}^{-1}\left[ \left\{
 b_{ij}\right\}\right]$ for some nonnegative integers $m, n_i$ $(1\le
 i\le m)$, some $y_{ij}\in \fld{\D}$, and some $b_{ij}\in \{0,1\}$
 ($1\le i\le m$, $1\le j\le n_i$.) So, define a Boolean formula over
 $\fld{\D}$ by $\bigvee_{i=1}^m\bigwedge_{j=1}^{n_i}
 (b_{ij}\leftrightarrow y_{ij})$, where each $b_{ij}\leftrightarrow y_{ij}$
 represents a Boolean formula $y_{ij}$ for $b_{ij}=1$ and the negation
 $\overline{y_{ij}}$ for $b_{ij}=0$. Clearly we have $\op(g)(x)=1$ iff
 $g\in \op^{-1}\left[\pi_x^{-1}[\{1\}]\right]$ iff 
 $g$ satisfies $B_x$.
\qed\end{proof}

For functions
$f,g\in \{0,1\}^Z$, we write $f\le g$ if $f(z)\le g(z)$ for all $z\in Z$.
\begin{definition}[Monotone functions]Let $\C\subseteq \P{X}$ and
 $\D\subseteq \P{Y}$.
We say a function $\op : \D \to \C$ is \emph{monotone},
if $f\le g$ implies $\op(f) \le \op(g)$.\end{definition} 

We say a \emph{Boolean formula} positive if it does not contain a
negation.

\begin{definition}
Let $\S$ and $\U$ be two (not necessarily distinct) sets of objects, and
$R$ be a $R\subseteq \S\times\inhfin{\U}$.
For $M\subseteq \U$ and $\M \subseteq P(\U)$, define
\begin{displaymath}
 R^{-1}[[M]]:= \left\{ s\;;\; \exists v\in \inhfin{M}.\  R(s,v)
 \right\},\ \ 
\
\widetilde{R^{-1}}[[\M]]:= \left\{ R^{-1}[[M]] \;;\; M\in \M \right\}.
\end{displaymath}
\label{def:bang}
Define $!\M:= \left\{
 \inhfin{M}\;;\; M\in \M\right\}$. Then $\bigcup !\M\subseteq \inhfin{\bigcup\M}$.
\end{definition}

\begin{lemma}\label{lem:or}
\begin{enumerate}
\item \label{thm:expcont} 
 Following conditions are equivalent:
\begin{enumerate}

\item A function $\op:\D\to \C$ is monotone and continuous.

\label{assert:poscon}

\item $\op$ is a function of $g\in\D$ and $x\in \fld{\C}$ such that it
      first produces a positive Boolean formula $B_x$ over
      $\fld\D$,  and then queries to an oracle
 $g$ whether $g$ satisfies $B_x$ or not.
\label{assert:posdnf}
\end{enumerate} 

\item \label{lem:Or}
 If
 $R\subseteq \fld{\C}\times \inhfin{\fld{\D}}$ is a finitely branching
 relation, then
\begin{equation}
 \op_R(g)(x):= \bigvee_{R(x,v)} \bigwedge_{y\in
 v}\left(g(y)=1\right),\quad(g\in \D,\  x\in \fld{\C}) \label{disjtt}
\end{equation} 
defines a monotone, continuous function from $\D$ to $\C$ such that,
\begin{equation}
 \widetilde{R^{-1}}[[\M]]\ =\  \inj^{-1}\left( \op_R\left[\inj
 \M\right]\right)\ \subseteq\ P\left(\fld{\C}\right),\qquad\left(\, \M \subseteq P(\fld{\D})\,.\right) \label{special}
\end{equation}
In fact, every monotone, continuous function from $\D$ to $\C$ is written as
\eqref{disjtt}.
\end{enumerate} 
\end{lemma}

\begin{proof}(1)  By Lemma~\ref{prop:b}. Positivity of a Boolean formula is equivalent
 to absence of negation in the formula. (2) follows from (1).
\qed\end{proof} 

\begin{lemma}\label{lem:bang} $\widetilde{R^{-1}}[[\M]]= \widetilde{R^{-1}}[!\M]$.\end{lemma}
\begin{proof}$L\in \widetilde{R^{-1}}[[\M]]$ iff there exists
 $M\in\M$ such that $L=R^{-1}[[M]]=\{x\;;\; \exists v\in \inhfin{M}.\
 R(x,v)\}=R^{-1}\left[ \inhfin{M}\right]\in \widetilde{R^{-1}}\left[\,!\M\right]$.\qed\end{proof}

\begin{theorem}\label{thm:bang} If $\M$ is an \textsc{fess}, so is
$!\M$.
\end{theorem}
\begin{proof} Otherwise there
exist an infinite sequence $v_0,v_1,\ldots$ of elements of $\bigcup
 !\M$ and an infinite sequence $\inhfin{M_1},
 \inhfin{M_2},\ldots$ of elements of $!\M$ such that for each $n\ge1$ we
 have $\{v_0,\ldots, v_{n-1}\}\subseteq\inhfin{M_n}\not\ni v_n$, which implies $\bigcup_{i=1}^{n-1}v_i\subseteq M_n\not\supseteq v_n$. Put
 $v_i':=v_i\setminus M_i$ ($i=0,1,\ldots$.) Then $v'_i\cap
 v'_j=\emptyset$ ($0\le i<j$) and each $v'_i$ is a nonempty finite
 set. Therefore $\{v'_i\;;\; i\in \Nset\}$ satisfies the Hall's condition
 of the marriage theorem~\cite[Theorem~3.41]{MR2440898}: for each finite set $F\subset\Nset$ we have $\#\left(
 \bigcup_{i\in F} v'_i \right)\ge \#F$. By the marriage theorem, $\{
 v'_i\;;\; i\in \Nset\}$ has a system of distinct representative
 $\{y_i\;;\;i\in\Nset\}$, i.e., $y_i\ne y_j$ ($0\le i<j$) and $y_i\in
 v'_i$ ($i=0,1,\ldots$.) Then for each $n\ge1$ $\{y_0,\ldots,
 y_{n-1}\}\subseteq \bigcup_{i=0}^{n-1}v_i \subseteq M_n$, while
 $y_n\not\in M_n$ because $y_n\in v'_n=v_n\setminus M_n$. This
 contradicts the \textsc{fe} of $\M$.\qed\end{proof}

The previous theorem generalizes Proposition~\ref{prop:5} which is
useful in Section~\ref{sec:cont_image}.

\begin{corollary}\label{thm:newkan}
Let $\M \subseteq P(\U)$ be an \textsc{fess} and let $R\subseteq \S\times
 \inhfin{\U}$ be a finitely branching relation. Then 
$\L= \widetilde{R^{-1}}[[\M]]\subseteq P(\S)$ is also an \textsc{fess}.
\end{corollary}
\begin{proof} By Theorem~\ref{thm:bang},  Lemma~\ref{lem:bang} and
Proposition~\ref{prop:5}. \qed
\end{proof}
Conversely, Theorem~\ref{thm:bang} follows from
Corollary~\ref{thm:newkan} with $\U:=\bigcup \M$,
$\S:=\inhfin{\bigcup\M}$, 
and a following:
\begin{definition}\label{def:rbang}
\begin{displaymath}
R_!:=\left\{(v,v)\;;\; v\in \inhfin{\bigcup\M}\right\},\ \mbox{and}\ \ 
 !\M=\widetilde{R_!^{-1}}[[ \M]]
\end{displaymath}
\end{definition}

%


In terms of topology, the previous corollary becomes a following:

\begin{corollary}\label{cor:contimage}Assume $\L$ and $\M$ are set systems
 and  $\op:\inj \M\to \inj \L$
 is a monotone, continuous function. Then if $\M$ is an \textsc{fess},
 so is
$\inj^{-1} \left( \op\left[\inj \M\right]\right)$.
\end{corollary} 

\begin{proof} By Lemma~\ref{lem:or}~\eqref{lem:Or}, 
$ \op(1_M)(x) = \bigvee_{R(x,v)} \bigwedge_{y\in v } (y\in M)$
where $R\subseteq \bigcup\L\times \inhfin{\bigcup \M}$ is a finitely branching relation. Therefore
we have
$L=R^{-1}[[M]]\iff L=\{x\;;\; \exists v\left( R(x,v)\ \&\ v\subseteq M\right)\} \iff
 1_L=\op (1_M)$. Hence the family $\inj^{-1}\left( \op\left[ \inj \M
 \right]\right)$ is $\widetilde{R^{-1}}[[\M]]$, which is an \textsc{fess} by Corollary~\ref{thm:newkan}.
\qed
\end{proof} 

Although the mind-change complexity of
language identification from positive data is characterized by using
the
\emph{\pit}~\cite{luo06:_mind_chang_effic_learn,MR2640836,brecht09:_topol_and_algeb_aspec_of},
Corollary~\ref{cor:contimage} does not hold for \pit.  Recall that the
\pit\ $\C\subseteq \{0,1\}^X$ is induced by the product
topology of the topology $\{0,1\}$ where the only nontrivial open subset
of $\{0,1\}$ is $\{1\}$. So the basic open sets of the \pit\ are
\begin{equation}
U^\C_F=\{ f\in \C ; f [ F] = \{1\}\}\  \mbox{where $F$ is an arbitrary
 finite subset of $X$}.  \label{u}
\end{equation} 
Let us abbreviate ``continuous with respect to the \pit'' by ``$\Pi$-continuous.''

\begin{lemma}A monotone, continuous function is $\Pi$-continuous.  
\end{lemma}
\begin{proof} Let $\op:\D\to\C$ be  a monotone, continuous function 
 and let $U_F^\C$ be a basic open set  of the \pit\ $\C$ where
$F$ is a finite subset of $\fld\C$.  By
 Lemma~\ref{lem:or}~\eqref{thm:expcont}, there are positive Boolean
 formulas $B_x$ over $\fld{D}$ ($x\in F$) such that for every $F\in
 \inhfin{\fld{\C}}$ the inverse image $\op^{-1}\left[ U^\C_F \right]$ 
is $\bigcap_{x \in F}\;\{ g\in \D\;;\; g\ \mbox{satisfies}\ B_x
\}$. Observe that each $B_x$ is equivalent to $\bigvee_{i=1}^{n_x}\bigwedge
 F_{x,i}$ for some $n_x\ge0$ and some $F_{x,i}\in \inhfin{\fld{\D}}$
 ($1\le i\le n_x$.) Therefore $\op^{-1}\left[ U^\C_F\right]$ is
 $\bigcap_{x\in F} \bigcup_{i=1}^{n_x} U^\D_{F_{x,i}}$, which is open
 with respect to \pit\ 
 because $F$ is finite. \qed \end{proof}

Recall that $\SGL=\{\{n\}\;;\;n\in\Nset \}$ is an \textsc{fess}. 
For $L\subseteq\Nset$, let $\downarrow L\subseteq \Nset$ be the downward
closure $\{n\;;\; n\le m\ (\exists m\in L)\}$ of $L$. To decide whether $n\in \downarrow L$, we must carry out
\emph{unbounded} search to find some $m \in L\cap [n,\,\infty)$.
\begin{theorem}\label{thm:7}
\begin{enumerate}
\item A function $\op_\downarrow\;:\; \P{\Nset}\to \P{\Nset}$ that sends
$1_L$ to $1_{\downarrow L}$ is monotone and $\Pi$-continuous but
 $\inj^{-1} \op_\downarrow \left[\,\inj \SGL\,\right]$ is not an
\textsc{fess}; and

\item There is a non-monotone, continuous, non-$\Pi$-continuous
function $\op_\neg:\P\Nset\to\P\Nset$  such that
 $\inj^{-1} \op_\neg \left[\,\inj \SGL\,\right]$ is not an
\textsc{fess}.
\end{enumerate}
\end{theorem}
\begin{proof}A basic open set \eqref{u} with $\C=\P\Nset$ is simply
 written $U_F$ below:
(1)
The monotonicity of $\op_{\downarrow}$ is obvious. The function $\op_\downarrow$ is $\Pi$-continuous, because for
      every basic open set $U_F$ with finite $F\subseteq \Nset$, the
      inverse image by $\op_\downarrow$ is an open set
$\bigcup_{B\subseteq F}\left(
      U_{F\setminus B} \cap \bigcap_{b\in B}\bigcup_{x>b}
      U_{\{x\}} \right)$ where $\bigcap_{b\in B}\cdots$ is $\P\Nset$ if
      $B=\emptyset$. However $\inj^{-1} \op_\downarrow\left[\, \inj\,{\SGL}\, \right] =
\DCL$ is not an
      \textsc{fess}.
\medskip

(2)
Moriyama-Sato~\cite{735102} observed that the elementwise complement
does not preserve the \textsc{fe} of set systems. Define
\begin{displaymath}
\op_\neg(g)(x)=1-g(x).
\end{displaymath} Then 
$\inj^{-1} \op_\neg[\,\inj \SGL\,]=\{ \Nset \setminus\{y\}\ ;\ y\in\Nset\}$ has an infinite
elasticity: $0,\ \Nset\setminus\{1\},\ 1,\ \Nset\setminus\{2\}, 2,\
\ldots.$ If $\op_\neg$ is $\Pi$-continuous, then 
 $\op_\neg^{-1}\left[ U_{\{0\}} \right]$ should be $\bigcup_G U_G$ where $G$ ranges over a certain class of finite
 subsets of $\Nset$. For such a finite set $G$, $g=1_G$ belongs to the
 inverse image by $\op_{\neg}$, but the support should be
 $\Nset\setminus \{0\}$. Contradiction.\qed
\end{proof}

\section{The order types of nondeterministically deformed set systems}
\label{sec:closureop}

 We present a typical application of
Corollary~\ref{cor:contimage}, and answer Question~\ref{q4} ``How much do such operations increase the order
type of set systems?'' by a Ramsey number argument.

Fix an alphabet $\Sigma$. To know whether a word $w$ belongs to the
Kleene closure 
$L^*=\bigcup_{n\ge0} L^n$ of a language $L$, we need to guess $n$ nondeterministically. Nondeterministic
operations such as the Kleene closure operator $(\cdot)^*$ and the
shuffle-closure operator $(\cdot)^\circledast$ are representable by
monotone, continuous functions.  So Corollary~\ref{cor:contimage} is useful
in deriving the following:
\begin{equation}
\M \subseteq P(\Sigma^*)\ \mbox{is an \textsc{fess}}\ \Rightarrow\
 \ewstar{\M}\ \mbox{and}\ \ewplus{\M}\  \mbox{are \textsc{fess}s}. \label{kleene}
\end{equation}

Let us see the proof to generalize for the case of the shuffle-closure.
Assume $\M$ is an \textsc{fess}. Let $\varepsilon$ be the empty word and let $\op_1
(1_M):=1_{M\setminus\{\varepsilon\}}$ and $\op_3
(1_M):=1_{M\cup\{\varepsilon\}}$. Then $\op_1$ and $\op_3$ are monotone
and continuous. Let $\op_2(1_L)$ be computed by a Turing machine with the
oracle tape being $1_L$ as follows: if an input $s\in\Sigma^*$ is
$\varepsilon$ then the oracle Turing machine returns 0. Otherwise, 
it tries to find a partition
$s_1,\ldots, s_m$
of $s$ such that $s=s_1\cdots s_m$,
$ln(s_i)>0$ $(1\le i\le m)$,
$m\ge1$ and $\{s_1,\ldots, s_m\}\subseteq L$. If such a partition is
found, then the oracle Turing machine returns 1, and 0 otherwise.
The number of queries the oracle Turing
machine makes is bounded by the number of partitions of $s$, which implies the continuity of $\op_2$. It is easy to see $\op_2$ is monotone.  Observe
 $\op_2(1_L)=
1_{L^+}$ for all $L\subseteq
\Sigma^*\setminus\{\varepsilon\}$. We can prove, for every
$M\subseteq\Sigma^*$, 
\begin{displaymath}
\left(M\setminus\{\varepsilon\}\right)^+ = M^+
\setminus\{\varepsilon\},\qquad \left( M\setminus \{\varepsilon
\}\right)^+\cup\{ \varepsilon \} = M^*.
\end{displaymath}
So we have
$\op_3\circ\op_2\circ\op_1(1_M)=1_{M^*}$. 
By Corollary~\ref{cor:contimage}, $\ewstar{\M}$ is an \textsc{fess}.

Assume $\ewplus{\M}$ has an infinite production sequence $\Langle \l t_0,
M_1^+ \r, \pair{t_1}{M_2^+}, \ldots\R$. Note that there is at most one $i$ such
that $t_i=\varepsilon$. Removal of such $\pair{t_i}{M_i}$ from the
infinite production sequence of $\ewplus{\M}$ results in still an infinite
production sequence of $\ewplus{\M}$. By adjoining the empty word
$\varepsilon$ to each language in the infinite production sequence, we
have an infinite production sequence of $\ewstar{\M}$, because
$M^+\cup\{\varepsilon\} = M^*$. But this is a contradiction against the
\textsc{fe} of $\ewstar{\M}$. So, $\ewplus{\M}$ is an \textsc{fess}.
\medskip

Remind that to find such a partition can be done by a nondeterministic
computation. 
We can prove the counterpart of \eqref{kleene} for the 
shuffle-closures $(\cdot)^\circledast$, as follows:

\begin{corollary}\label{cor:shuffle}
If $\L \subseteq P(\Sigma^*)$ is an \textsc{fess}, so is $\L^{\widetilde{\circledast}}$.
\end{corollary}

\begin{proof} The proof is similar to that of
 \eqref{kleene} except $\op_2(1_L)$ is computed by another Turing
 machine with the oracle tape being $1_L$ as follows: if an input $s\in
\Sigma^*$ is $\varepsilon$, then it returns 0.  Otherwise, it tries to find a sequence $s_1,\ldots, s_m$ such that $s$ is an ``interleaving
 merge'' of $s_1,\ldots, s_m$,  $ln(s_i)>0$ $(1\le i\le m)$, $m\ge1$ and
 $\{s_1,\ldots, s_m\}\subseteq L$. Then $\op_2$ is clearly monotone and continuous. Moreover $\op_2(1_L)= 1_{L^\diamond}$ for every
 $L\subseteq \Sigma^*\setminus\{\varepsilon\}$.  We can prove, for every
 $M\subseteq\Sigma^*$, 
\begin{displaymath}
\left(M\setminus\{\varepsilon\}\right)^\diamond = M^\diamond
\setminus\{\varepsilon\}, \qquad \left( M\setminus \{\varepsilon
\}\right)^\diamond\cup\{ \varepsilon \} = M^\circledast.
\end{displaymath}   So we have
$\op_3\circ \op_2\circ\op_1(1_M)=1_{M^\circledast}$. 
By Corollary~\ref{cor:contimage}, we have done. \qed
\end{proof}

 We can see that \eqref{kleene} also holds for tree languages~\cite{267871}.

\bigskip
Next we answer Question~\ref{q4} ``How much do such operations increase
the order type of set systems?'' by a Ramsey number argument.

The finitely branching relation 
\begin{equation}R_n(s,u):\iff \bigvee_{i=0}^{n-1} \ \left(u=\left\{\langle
 s,i\rangle\right\}. \right)\qquad (n=2,3,\ldots)  \label{prod_to_ewunion}
\end{equation}
satisfies
$\widetilde{R_n^{-1}}\left[[ \L_1\ewuplus\cdots \ewuplus \L_n ]\right] = \L_1 \ewunion
 \cdots\ewunion \L_n$.
So, if $\L_i$'s are all \textsc{fess}s, then so is $\L_1\ewuplus\cdots
\ewuplus \L_n$ by Lemma~\ref{lem:a}. By Corollary~\ref{thm:newkan}, $\L_1 \ewunion
 \cdots\ewunion \L_n$ is an \textsc{fess}, too.

On the other hand, in \cite{93373}, Wright proved that ``if $\L_1$ and $\L_2$ are
\textsc{fess}s, then so is $\L_1\ewunion\L_2$,'' by using Ramsey
theorem ``for any dichromatic coloring of an infinite complete graph,
there is a monochromatic infinite complete subgraph.'' By adapting his
proof, we can provide an explicit upper bound of the dimension  by using
a 
\emph{Ramsey number}~\cite{MR1044995}:
\begin{proposition}[\protect{\cite[Sect~1.1]{MR1044995}}]
 For all positive integers $n, l_1,\ldots,l_n$, there exists a positive integer $k$ such
 that any edge-coloring with colors $1,\ldots,n$ for the complete graph of
 size $k$
 has a  complete subgraph of size $l_i$ colored homogeneously by some
 color $i\in \{1,\ldots,n\}$ . Such minimum integer $k$, denoted by $\Ram(l_1,\ldots,l_n,n)$, is called the \emph{Ramsey
 number of $l_1,\ldots,l_n$}. When $l_1=\cdots=l_n$, we call it the \emph{$n$-adic
 diagonal Ramsey number of $l_1$}, and 
write it as
 $\Ram(l_1; n)$.
For the sake of convenience, put $\Ram(l;1)=l$ for every nonzero ordinal number
 (and hence every positive integer) $l$.
\end{proposition}
By \cite[Section~4.2]{Graham.Rothschild.ea:80},
$\Ram(l,m)\le
\left(\begin{array}{l}{m+l-2}\\{l-1}\end{array}\right)\le {c4^{\max(l,m)}}/{\sqrt{\max(l,m)}}$ for some constant $c$.

\begin{lemma}\label{thm:ramseyfn} For every positive integer $n$, if $\dim \L_i<\omega$ $(i=1,\ldots,n)$, then
\[
\dim\left(\L_1 \ewunion \cdots \ewunion \L_n\right) +1 <
 \Ram(\dim(\L_1)+2,\ldots, \dim(\L_n)+2).
\]
\end{lemma}

\begin{proof}When $n=1$ the assertion is trivial.
Consider the case $n=2$.
  Suppose that $k + 1 \ge 
\Ram(\dim(\L)+2,\, \dim(\M)+2)$, and
suppose there are a sequence $t_0,\ldots,t_{k-1}$, a sequence $L_1,\ldots,
L_k$ of $\L$
and a sequence $M_1,\ldots, M_k$ of $\M$ such that
\begin{equation}
 \{t_0,\ldots,t_{i-1} \}\subseteq L_i\cup M_i\
  (i=1,2,\ldots,k)\quad\mbox{and}\quad L_j \not \ni
t_j\quad(j=1,\ldots k-1.)
\label{Nseq:dag}
\end{equation}

By the definition of Ramsey number, 
\begin{displaymath}
k \ge \dim (\L) + 1, \dim (\M) +1.  
\end{displaymath}

Consider a complete graph $G$ with the vertices being ${0,\ldots, k-1}$.
For any edge $\{i,j\} (i\ne j)$,
color it by red if $0\le i<j\le k-1$ and $t_i\in L_{j}$,
while color it by black otherwise.  

Assume $k + 1\ge \Ram(\dim(\L)+2,\ \dim(\M)+2)$. By Ramsey's theorem, the
colored complete graph $G$ has either a red clique of size $\dim(\L)+2$ or
a black clique of size $\dim(\M)+2$. When a red clique of size
 $\dim(\L)+2$ exists, write it
 as $\{u_0<\cdots< u_{\dim(\L)+1}\}$. Then we have
$\{t_{u_0},\ldots,t_{u_{i-1}}\}\subseteq L_{u_i}$ 
$(i=1,\ldots, \dim(\L)+1)$ but $L_{u_j}
\not \ni
 t_{u_j}\quad(j=1,2,\ldots, \dim(\L))$,
which contradicts the definition of $\dim \L $. 

Otherwise, a black clique
 of size $\dim(\M)+2$ exists, so we
write it
 as $\{u_0<\cdots<u_{\dim(\M)+1}\}$. Then we have
$\{t_{u_0},\ldots,t_{u_{i-1}}\}\cap L_{u_i}=\emptyset\quad(
 i=1,\ldots, \dim(\M)+1)$.
By \eqref{Nseq:dag}, we
 have 
$\{t_{u_0},\ldots,t_{u_{i-1}}\}\subseteq
 L_{u_i}\cup M_{u_i}$ and $L_{u_j}\cup M_{u_j} \not\ni t_{u_j}$
 $(j=1,\ldots, \dim (\M))$, so 
$\{t_{u_0},\ldots,t_{u_{i-1}}\}\subseteq  M_{u_i}$ and $M_{u_j} \not\ni
 t_{u_j}$ $(j=1,\ldots, \dim\M)$,
which contradicts the definition of $\dim(\M)$.  

Consider the case $n\ge3$. Suppose that $k+1\ge \Ram(\dim \L_1+2,\ldots,
 \dim \L_n+2)$, and suppose there are a sequence $t_0,\ldots,
 t_{k-1}$ and a sequence $L_1^{(l)},\ldots, L_k^{(l)}$ of $\L_l$ ($1\le
 l\le n$) such that $\{t_0,\ldots, t_{i-1}\}\subseteq \bigcup_{l=1}^n
 L_i^{(l)}$ $(i=1,\ldots,k)$ and $\bigcup_{l=1}^n L_j^{(l)}\not\ni t_j$ ($j=1,\ldots,k-1$.)
For any edge $\{i,j\}$ $(i\ne j)$, if $0\le
 i<j\le k-1$ and $t_i\in L_j^{(1)}$, color $\{i,j\}$ by the color 1;  else if 
 $0\le i<j\le k-1$ and $t_i\in L_j^{(2)}$, color it by the color 2; else
 if $\ldots$; else if $0\le i<j\le k-1$ and $t_i\in L_j^{(n-1)}$,
 color it by the color $n-1$;
 else color $\{i,j\}$ by the color $n$. Then apply
 the same argument as above.
\qed
\end{proof}

The lemma generalizes for any relation $R$ with  $\sup_s \#\{ v\;;\;
R(s,v)\}\le n$. 
\begin{theorem}\label{cor:akama}
Assume $R\subseteq \bigcup\L \times \inhfin{\bigcup\M}$ has a bound $n\ge 1$ of $\#\{ v\;;\; R(x,v)\}$ $(x\in\bigcup\L)$. If $\M$ is an
 \textsc{fess}, then
\begin{displaymath}
 \dim \widetilde{R^{-1}}[[\M]] + 1 <  \Ram(\dim \M+2; n),
\end{displaymath}
provided $\dim \M$ is finite or $n=1$. Actually, when $n=1$, 
\begin{equation}
\dim\ \widetilde{R^{-1}}\left[[\M]\right] \le \dim \ \M  \label{n1}
\end{equation}where
the equality holds if  each $y\in \bigcup\M$ has $\xi(y)\in \bigcup\L$
 such that $R(\xi(y),\{y\})$. 
\end{theorem} 

\begin{proof}
To show the inequality for $n=1$,  by Fact~\ref{fact:TE}, it is
 sufficient to build an embedding $f$ from
a well-founded tree $\Prd{\widetilde{R^{-1}}[[\M]]}$ to a well-founded tree $\Prd{ \M }$.
Suppose
\begin{displaymath}
a=\Langle
 \pair{s_0}{L_1},\ldots, \pair{s_{l-1}}{L_l}\R\in
 \Prd{\widetilde{R^{-1}}[[\M]]}.
\end{displaymath} 
 For each $L\in \widetilde{R^{-1}}[[\M]]$, choose  $M(L)$ from $\left\{M\in
 \M\ ;\ L=R^{-1}[[M]]\right\}\ne\emptyset$.
For each $i=0,\ldots,l-1$, because $n=1$, there exists exactly one $v_i$ such that
 $R(s_i,v_i)$ and $v_i \in \inhfin{M(L_{i+1})}$. Since $s_i\not\in L_i$, 
 $v_i\not\subseteq M(L_i)$. Because the class of finite sets $v_i\setminus
 M(L_i)$ satisfies the Hall's condition of the marriage~\cite[Theorem~3.41]{MR2440898} theorem, we have a
 system $\{y_i\;;\;i\in \Nset\}$ of distinct representative.
 Obviously $y_i\not\in M(L_i)$.
Define
\begin{displaymath}
f(a):=\Langle \pair{y_0}{M(L_1)}\,  ,\ldots, \pair{y_{l-1}}{M(L_l)}\,\R.
\end{displaymath}
 We have indeed $f(a) \in \Prd{\M}$, because for $0\le i<j\le l$, since
 $s_i\in L_j$, $y_i\in v_i\subseteq M(L_j)$. 
The mapping $f$ is indeed injective by the
 construction. Clearly $f$ preserves the greatest upper bounds.

The verification of the \emph{equality} is as follows:
By $n=1$ and the assumption of Theorem~\ref{cor:akama}, we have $\exists
 v\in \inhfin{M_i}.\ R(\xi(y_j), v)\iff \{y_j\}=v\subseteq M_i$. So
 $\xi$ is injective. Moreover \begin{equation}
\xi(y_j)\in  R^{-1}[[M_i]] \iff y_j\in M_i. \label{hosihosi}
\end{equation}
Define a function $g$ as:
\begin{eqnarray*}
& b=\Langle \pair{y_0}{M_1},\ldots,\pair{ y_{l-1}}{ M_l}\R \in \Prd{ \M }\\
\mapsto\ &g(b):=  \left\langle \langle \xi(y_0),  R^{-1}[[M_1]]\,\rangle,\ldots,\pair{ \xi(y_{l-1})}{ R^{-1}[[M_l]]}\R. 
\end{eqnarray*}
Then $g(b)\in \Prd{ \widetilde{R^{-1}}\left[[\M]\right] }$ by \eqref{hosihosi}. 
The injectivity of $g$  is from that of $\xi$.
The preservation of glb's by $g$  is easy.

Next we prove the case where $n>1$ and $\dim \M<\omega$.
There are  relations $R_i\subseteq \bigcup \L\times \inhfin{\bigcup \M}$
 such that for all $x\in \bigcup\L$ and all $v\in \inhfin{\bigcup\M}$
\begin{displaymath}
 R=\bigcup_{i=1}^n R_i\ \ \mbox{and}\ \ \#\{
 v\;;\; R_i(x,v)\}\le 1.
\end{displaymath}
 Then for all $M\in\M$, we have
$R^{-1}[[M]]= \bigcup_{i=1}^n  R_i^{-1}[[M]]$,
 because the left-hand side is
$\{ s\in \bigcup\L\;;\; \exists v\subseteq M.\ R(s,v) \}
         = \bigcup_{i=1}^n \ \{ s\in \bigcup\L\;;\; \exists
 v\subseteq M.\ R_i(s,v) \}$ which is the right-hand side. So
we have 
\begin{displaymath}
\widetilde{R^{-1}}[[\M]]\subseteq \widetilde{R_1^{-1}}[[\M]]\ewunion\cdots\ewunion
\widetilde{ R_n^{-1}}[[\M]].
\end{displaymath} By
 Lemma~\ref{thm:ramseyfn},
 we have 
\begin{displaymath}
\dim \widetilde{R^{-1}}[[\M]]+1\ \le\ \Ram(\dim \widetilde{R_1^{-1}}[[\M]]+2,\ldots,\dim
\widetilde{ R_n^{-1}}[[\M]]+2).
\end{displaymath} 
Since we have already proved \eqref{n1}, we can use \eqref{n1} to derive
 $\dim \widetilde{R_i^{-1}}[[\M]] \le \dim \M < \omega$.  The monotonicity
 of $\Ram$ concludes the desired consequence. \qed
\end{proof}

We will use Theorem~\ref{cor:akama} again to derive the linearization
~(Corollary~\ref{cor:linearbound}) of \textsc{wqo}s.

Let us see an example of the inequality \eqref{n1} of Theorem~\ref{cor:akama}.

\begin{corollary}\label{ew:cap}
Let $\L$ and $\M$ be \textsc{fess}s.
\begin{enumerate}
\item \label{ew:cap:cap}
Let $R$ be $\bigl\{\;\left(s, \left\{\pair{s}{s}\right\}\right)\;;\; s\in
 \bigcup \L \cap \bigcup \M\bigr\}$.
\begin{enumerate}
\item \label{assrt:cap:1} $\widetilde{R^{-1}}\left[[\L\ewprod\M]\right]
      =\L\ewcap \M $. 
\item \label{assrt:2} $\L\ewprod\M$ is an \textsc{fess}.

\item \label{assrt:3}
If $\dim \L<\omega$ and $\dim\M<\omega$, then 
\begin{displaymath}
 \dim(\L\ewprod\M)\ge\dim\L + \dim\M -1\ge \dim(\L\ewcap \M).
\end{displaymath} 
The inequalities are best possible.
\end{enumerate} 

\item \label{ew:cap:bang}$\dim \M=\dim\; !\M$.
\end{enumerate} 
\end{corollary}
\begin{proof}\eqref{assrt:cap:1} is immediate. To prove \eqref{assrt:2},
assume $\L \ewprod \M$ is not an \textsc{fess}. Then  we have an infinite production sequence
\begin{displaymath}
\Bigl\langle \l (t_0,p_0), L_1\times M_1\r,\ \l(t_1,p_1),\, L_2\times M_2\r,
  \ldots \Bigr\rangle\in \Prd{\L \ewprod \M}.
\end{displaymath}
 When $k:=\sup \{i\;;\;
 L_i\not\ni t_i\}<\infty$, then for all $i>k$, we have $M_i\not\ni
 p_i$, and thus an infinite production sequence
$\Bigl\langle \l p_k, M_{k+1}\r , \l p_{k+1},  M_{k+2}\r,   \ldots \Bigr\rangle$ of
 $\M$, contradicting the \textsc{fe} of $\M$. Otherwise, we have an
 infinite sequence $i_0, i_1, \ldots$ such that 
$\Bigl\langle \l t_{i_0}, L_{i_1}\r,\ \l t_{i_1},\  L_{i_2}\r, \ldots \Bigr\rangle\in \Prd{\L}$,  contradicting the
 \textsc{fe} of $\L$. 

\medskip
To show  $\dim (\L \ewprod \M)\ge \dim \L + \dim \M -1$ of \eqref{assrt:3},
let 
\begin{eqnarray*}
&\Bigl\langle \l t_0, L_1\r,\ \l t_1, L_2\r,\ldots, \l t_{l-1},L_l\r\Bigr\rangle\in
 \Prd{\L},\\
&\Bigl\langle \l p_0, M_1\r, \l p_1, M_2\r,\ldots,
\l p_{m-1},M_m\r\Bigr\rangle\in\Prd{\M}.
\end{eqnarray*} Then the class $\L \ewprod \M$ has a following
production sequence consisting of $(l+m-1)$ members of $\L \ewprod \M$:
\[
\begin{array}{llllllllllllllll}
\Bigl\l \l (t_0,p_0), L_1\times M_1\r,  \qquad\l (t_0,p_1), L_1\times M_2\r,
\ \  \ldots,\ \ \l (t_0,p_{m-1}),\  L_1\times M_m\r,\\
\ \ \l (t_1,p_{m-1}), L_2\times M_m \r,\ \l (t_2,p_{m-1}), L_3\times M_m \r, \ldots, \l (t_{l-1},p_{m-1}), L_l\times M_m\r\ \Bigr\rangle.
\end{array}
\]
Thus $l+m-1\le \dim (\L\ewprod\M)$. Since $\dim \L<\omega$ and
 $\dim\M<\omega$, we have $\dim\L + \dim \M -1\le \dim (\L\ewprod\M)$.
The equality is attained by $\L=\M=\{\{1\}\}$.

\medskip
To verify the inequality $\dim \L + \dim \M -1\ge \dim(\L\ewcap\M)$   of \eqref{assrt:3},
 let 
\begin{displaymath}
\Bigl\langle \langle t_0, L_1\cap M_1\r,\ \l t_1,  L_2\cap M_2\r, 
 \ldots, \l t_{n-1}, L_n\cap M_n\rangle\Bigr\rangle\in \Prd{\L \ewcap \M}.
\end{displaymath} Then $t_i\in \bigcup\L\;\cap\; \bigcup\M$, and
for every positive integer $i\le n$, we have
$\left\{ t_0,\ldots, t_{i-1}\right\}\subseteq L_i\cap M_i\not\ni
t_i$. So for each positive $i\le n-1$, $L_i\not\ni t_i$ or $M_i\not\ni t_i$. Let
$i_1,\ldots,i_l$ be the strictly ascending list of positive integers $i$ such
 that $L_i\not\ni t_i$, and $j_1,\ldots,j_m$ be the strictly ascending list of positive integers $j$ such
 that $M_j\not\ni t_j$.  Then $i_l\ne n$, $j_m\ne n$, and so
 $\Bigl\langle \l t_0, L_{i_1}\r, \l t_{i_1}, L_{i_2}\r,\ldots, \l
 t_{i_{l-1}},  L_{i_l}\r, \l t_{i_l}, L_n\r \Bigr\rangle\in \Prd{\L}$,  and
 $\Bigl\langle \l t_0, M_{j_1}\r, \l t_{j_1}, M_{j_2}\r,\ldots, \l
 t_{j_{m-1}},  M_{j_m} \r, \l t_{j_m}, M_n\r\Bigr\rangle\in \Prd{
 \M}$. Therefore $l+1\le \dim \L$ as well as $m+1\le \dim \M$. Because $n-1\le l+m
$,  we have $\dim (\L\ewcap \M) -1 \le (\dim \L -1) + (\dim \M -1)$, from
 which the conclusion follows.
The latter inequality of Corollary~\ref{ew:cap}~\eqref{assrt:3} is best
possible. The equality holds for 
\begin{equation}
\L=\{\emptyset, \{0\}, \{0,1,2\}\}\ \mbox{and}\ 
\M=\{\emptyset, \{1\}, \{0,1,2\}\},\label{eg:best}\end{equation}
because
\begin{equation}
 \L\ewcap \M=\{\emptyset, \{0\}, \{1\}, \{0,1,2\}\},\ 
  \dim(\L\ewcap\M)=3,\ \dim\L=\dim\M=2.
\label{why}
\end{equation}

 The assertion \eqref{ew:cap:bang} holds, because of Definition~\ref{def:rbang}, $n=1$, and $\xi(y)=\{y\}$. 
\qed
\end{proof}

There are many equivalent definitions of \textsc{wqo}s~(see
\cite[Theorem~2.1]{MR0049867} and \cite{MR0447056}.) 
In \cite{MR2078917},  
Cholak-Marcone-Solomon studied  for which definition
of \textsc{wqo} and which subsystem of second order
arithmetic~\cite{MR1723993} proves
\begin{displaymath}
\X\ \mbox{and}\ \Y\ \mbox{are \textsc{wqo}s} \Rightarrow
\X\cap \Y\mbox{ and } \X \times \Y\ \mbox{are \textsc{wqo}s.} 
\end{displaymath}
 The results are certainly related to a question ``for which
ordinal number do we have $\otp(\X), \otp(\Y)<\alpha\Rightarrow
\otp(\X\cap \Y)<\alpha$?'' We conjecture that we can take as $\alpha$ the proof-theoretic
ordinal $\Gamma_0$. According to Simpson~\cite[Ch.~V]{MR1723993},
$\Gamma_0$ is the proof-theoretic ordinal of a formal system which
can formalize and develop significant parts of order (type) theory. I
wonder whether we can take as $\alpha$ the the first nonrecursive ordinal.
If $\dim (\L)$ were almost equal to $\otp (\qo{\L})$ (cf. Remark~\ref{rem:false}),
then we would smoothly study which ordinal numbers satisfy
\begin{displaymath}
 \dim(\L), \dim(\M)<\alpha\implies  \dim(\L
 \anonbin \M)<\alpha,\quad (\anonbin=\ewprod,\ewunion,\ewcap,\ldots .)
 \end{displaymath}

A Ramsey number argument used in the proof of Lemma~\ref{thm:ramseyfn}
establishes an upper bound of a \textsc{wqo} obtained as the intersection of \textsc{wqo}s.

\begin{theorem}\label{thm:wpo} 
\begin{displaymath}
\otp(\X), \otp(\Y)<\omega\ \Rightarrow\ \otp\left(\X\cap \Y\right)<
 \Ram(\otp(\X)+1,\, \otp(\Y)+1, \, 2).
%
\end{displaymath}
\end{theorem}
\begin{proof}
%
The proof is similar as that of Lemma~\ref{thm:ramseyfn}.
Assume $\X=(X,\preceq)$, $\Y=(Y,\sqsubseteq)$ and $\langle
 t_1,t_2,\ldots,t_m\rangle$ is a bad sequence of $\X\cap \Y$. Then for
 all $i,j$ with $1\le i<j\le m$, we have $t_i\not\preceq t_j$ or
 $t_i\not\sqsubseteq t_j$.  For the complete graph consisting of
 $\{1,\ldots,m\}$, color all edges $\{i,j\}$ $(i\ne j)$ by red if
 $t_i\not\preceq t_j$, and color the other edges by black. Then there is
 a red complete graph consisting of size $\otp(\X)+1$, or
 a black complete graph of size $\otp(\Y)+1$. For the former case, the
 bad sequence $\langle t_1,t_2,\ldots,t_m\rangle$ has a bad subsequence, which consists of terms with the suffixes from the
 red graph's vertices. This bad sequence of $\X$ has the length
 $\otp(\X)+1$, a contradiction.
For the latter case, the black complete graph of size $\otp(\Y)+1$
 induces a bad subsequence of $\Y$ having the length $\otp(\Y)+1$, a
 contradiction. Thus, we have the desired consequence.
\qed
\end{proof}

One may conjecture
\begin{equation}
 \closure{\X\cap \Y} \subset \closure{\X}\ewcap\closure{\Y} \label{cain}
\end{equation} 
in order to derive a following asymptotic improvement of Theorem~\ref{thm:wpo}
\begin{equation}\otp(\X\cap\Y)< \otp(\X)+\otp(\Y)\ \mbox{for}\ \otp (\X),
\otp(\Y)<\omega,\label{eq:false}
\end{equation}
with an argument below: By \eqref{cain} and Theorem~\ref{thm:repre}~\eqref{assert:otp},
we have $\otp(\X\cap\Y)\le \dim\left(
\closure{\X}\ewcap\closure{\Y}\right)$, but Corollary~\ref{ew:cap}~\eqref{assrt:3}
implies the latter is less than or equal to 
$\dim \closure{\X}+\dim \closure{\Y} -1 =\otp(\X) + \otp( \Y) -1$.

However the inclusion of \eqref{cain} is actually opposite, when
$\X,\Y$ are following \textsc{wqo}s $\le_0$ and $\le_1$. Let $\le_i$ $(i=0,1)$ be a
\textsc{wqo} over $\{0,1,2\}$  such that the pair of $\closure{\le_i}$ $(i=1,2)$
is the pair of $\L$ and $\M$ presented in \eqref{eg:best}, which attains $\dim \L + \dim\M -1= \dim (\L\ewcap\M)$.
 Namely, $\le_i$ is such that
two
elements other than $i$ are mutually related by $\le_i$ and are strictly
lower than $i$ by $\le_i$. 
Then
$\le_0\cap\le_1$ becomes a \textsc{wqo} such that the elements 0 and 1 are
not comparable but are  strictly greater than the element 2. Thus
$\closure{\le_0\cap\le_1}= \{ \emptyset, \{0\}, \{1\}, \{0,1\},
\{0,1,2\}\}\supset \closure{\le_1}\ewcap \closure{\le_2}= \{ \emptyset, \{0\}, \{1\}, 
\{0,1,2\}\}$.

\section{Embedding of the category of quasi-orders and finitely
  branching simulations}\label{sec:embed}

In hope that we could import idea and results on closure properties of
\textsc{wqo}s and \textsc{bqo}s to study those of
\textsc{fe}s, we show that $\closure{\bullet}$ studied in
Section~\ref{sec:intersection} becomes a neat embedding from
the category of quasi-orders and \emph{finitely branching simulations} to the
category of set systems and \emph{linear monotone, continuous functions}.
Here a
``\emph{simulation}'' is used widely in theoretical computer science~(see \cite{Sangiorgi:2009:OBC:1516507.1516510}.)  ``\emph{Linear}''
is used in the model theory of linear logic~\cite{Gir89} and we will point
out that it corresponds
to Kanazawa's relation $R\subseteq X\times Y$~(see Proposition~\ref{prop:5}.)

By ``neat embedding,'' we mean that $\closure{\bullet}$ not only
 preserves order types but also, in the jargon of category
 theory~\cite{MR1712872}, becomes an injective-on-objects, full and
 faithful contravariant functor right adjoint to a functor that
 $\qo{\bullet}$ (see Section~\ref{sec:intersection}) induces.

\begin{definition}[Finitely Branching Simulation]
Let $\X=(X,\preceq)$ and $\Y=(Y,\sqsubseteq)$  be quasi-orders. We say a relation  $R$ is a \emph{simulation} of $\X$ by $\Y$,
provided
$R\subseteq X\times Y$  and  whenever $R(x,y)$ and $x\preceq x'$,  there
 exists $y'\sqsupseteq y$ such that $R(x',y')$. We say a simulation $R$
 is \emph{finitely branching} if $\#\{ y\;;\; R(x,y)\}<\infty$ for every $x$.
\end{definition}

\begin{example}[Lineariztion]\label{eg:simulation}\begin{enumerate}\rm
\item \label{assert:functional_simulation}
For an order-homomorphism $f:\X\to\Y$, a relation 
$R_f:=\left\{(x,y)\;;\; f(x)=y\right\}$ is a finitely branching simulation of $\X$ by $\Y$. 

\item For every surjective order-homomorphism $f$ from a quasi-order $\X$
to a linear order $\Y$, the relation $R_f$ is a finitely branching simulation of $\X$ by $\Y$. In
this case, we call $\Y$ a \emph{linearization} of $\X$.
\end{enumerate} 
\end{example}

\begin{lemma}\label{lem:simulation}Let $\X$ and $\Y$ be quasi-orders. If $R$  is a simulation of $\X$ by $\Y$,  then
$\widetilde {R^{-1}}[\closure{\Y}]\subseteq\closure{\X}$.
\end{lemma}
 \begin{proof}Let $\X=(X,\preceq)$ and $\Y=(Y,\sqsubseteq)$. 
Any member of $\widetilde {R^{-1}}[\closure{\Y}]$ is written as a set
 $L:=\{ x \in X\;;\; \exists g\in M.\, \exists y\sqsupseteq g.\ R(x,y)
 \}$ for some $M\in\closure{\Y}$. Suppose $x'\succeq x\in L$. Then
 because $R$ is a simulation, there is $y'$ such that $y'\sqsupseteq y$
 and $R(x',y')$. By the transitivity of $\sqsubseteq$, we have
 $y'\sqsupseteq g$. Therefore $x'\in L$. Thus $L$ is an upper-closed
 set, which implies $L\in \closure{\X}$.\qed
 \end{proof}

By Lemma~\ref{lem:simulation}, we have a so-called linearization lower bound~\cite[Sect.~2]{MR1283862}:
\begin{corollary}\label{cor:linearbound}For any linearization $\Y$ of $\X$,
 $\otp(\Y)\le \otp(\X)$.
\end{corollary}
\begin{proof} By the premise, there is a surjective order-homomorphism
 $f:\X\to\Y$. Because a relation $R_f$ is a simulation,
 Lemma~\ref{lem:simulation} implies 
$\dim \closure{\X}\ge \dim
 \widetilde{R_f^{-1}}[\closure{\Y}]$ from which
 Theorem~\ref{thm:repre}~\eqref{assert:otp} implies
\begin{equation} \otp ( \X) \ge \dim\widetilde{R_f^{-1}}[\closure{\Y}].
\label{twostar}\end{equation}

Put
\begin{displaymath}
Q_f:=\bigl\{(x, \{y\})\;;\;
 f(x)=y\bigr\}\subseteq\bigcup\closure{\X}\times\inhfin{\bigcup
 \closure{\Y}}.
\end{displaymath}
Then $\#\{v\;;\;Q_f(s,v)\}=\#\{v\;;\;\{f(s)\}=v\}\le 1$ and
\begin{equation}
 \dim \widetilde{R_f^{-1}} [\closure{\Y}] = \dim \widetilde{Q_f^{-1}}
  [[\closure{\Y}]],\label{hosihitotu}
\end{equation}
 because  $\widetilde{R_f^{-1}}[\closure{\Y}]= \{ R_f^{-1}[M]\;;\; M\in
 \closure{\Y}\} = \bigl\{ \{ x\in \bigcup{\closure{\X}}\;;\; \exists y\in
 M.\, f(x)=y\}\ ; \ M\in \closure{\Y}\bigr\} = \bigl\{ \{x\in\bigcup\closure{\X}\;;\;
 \exists v\in \inhfin{M}.\ Q_f(x,v)\}\ ;\ M\in \closure{\Y}\bigr\} = \{
 Q_f^{-1}[[ M]]\;;\; M\in\closure{\Y}\} = \widetilde{Q_f^{-1}}[[
 \closure{\Y}]]$.

Since $f:\X\to\Y$ is surjective and $\bigcup\closure{\Y}$ is the
 underlying set of $\Y$, there is a right-inverse $\xi:\bigcup\closure{\Y}\to\bigcup\closure{\X}$ of $f$. In other
 words, each $y\in \bigcup\closure{\Y}$ has $\xi(y)\in \bigcup\closure{\X}$
 such that $f(\xi(y))=y$, i.e., $R_f\left(\xi(y), y\right)$. Hence
 $Q_f\left(\xi(y),\,\{y\}\right)$. By 
 Theorem~\ref{cor:akama}, $\dim \widetilde{Q_f^{-1}} [[ \closure{\Y}]]=
 \dim \closure{\Y}=\otp{(\Y)}$. By \eqref{hosihitotu}, $\dim
 \widetilde{R_f^{-1}}[ \closure{\Y}] = \otp{(\Y)}$. By \eqref{twostar},
 we have the desired consequence.
\qed\end{proof}

We
will define the category $\QO_{FinSim}$ of quasi-orders and
finitely branching simulations between them, as well as a
suitable category of set systems and monotone, continuous functions
between them, and then will show that the operation $\closure{\bullet}$ becomes
 a contravariant, functor from the former category $\QO_{FinSim}$ to the
 latter category, and that the functor $\closure{\bullet}$ is
order-type-preserving,
injective-on-objects,  full and faithful.
For notion of category theory, see \cite{MR1712872}.

\begin{definition} \emph{The category $\QO_{FinSim}$ of quasi-orders and
finitely branching simulations between them} is defined as follows: The
objects are quasi-orders $(X,  \preceq)$. The identity
 morphism of object $\X=(X, \preceq)$ is $\id_\X=\{(x,x)\;;\; x\in
 X\}$. The morphisms
 from $\X=(X,\preceq)$ to $\Y=(Y,\sqsubseteq)$ are finitely branching simulations $R\subseteq X\times Y$.
 For morphisms $R:(X,\preceq)\to(Y,\sqsubseteq)$ and
 $S:(Y,\sqsubseteq)\to(Z, \trianglelefteq)$,
 the composition is defined as
 the relational composition 
\begin{displaymath}
S\circ R=\{(x,z)\;;\; R(x,y)\ \mbox{and}\ S(y,z)\
 \mbox{for some $y\in Y$}\}.
\end{displaymath}

Let $\QO$ be the category of quasi-ordered sets and order-homomorphisms between them. Then there is a faithful,
 identity-on-objects functor from $\QO$ to $\QO_{FinSim}$, because of Example~\ref{eg:simulation}~\eqref{assert:functional_simulation}.
\end{definition}

\begin{definition}[Linear, Sequential]Let $\D$ and $\C$ be set systems and  $\op:\D\to\C$ be a monotone, continuous function.
$\op$ is said to be \emph{linear}, if there is
 $R\subseteq\fld{\C}\times\left[ \fld{\D}\right]^{<2}$ such that
 $\op=\op_R$.
$\op$ is said to be \emph{sequential}, if there is
 $R\subseteq\fld{\C}\times\inhfin{\fld{\D}}$ such that
 $\op=\op_R$ and $\#\{ v\;;\; R(s,v)\}\le 1$ for all $s\in\fld{\C}$.
\bigskip

Let \emph{$\SetSys$ be the category of set systems and
 monotone, continuous functions between them}.
Let $\SetSys_{lin}$~(\/$\SetSys_{seq}$, resp.) be the category of set systems and linear~(sequential, resp.) monotone, continuous
 functions between them.
\end{definition}

Thus every object $\C$ of $\SetSys$ is written as $\inj \L$ for some set
system $\L$.

Let $\coh_{stable}$  be the cartesian closed category of
 coherence spaces and stable functions between them, introduced by
 Girard~\cite{Gir89}. Here a stable function was originally introduced
 by Berry in an attempt to give a semantic characterization of
 \emph{sequential} algorithms. 
Defining coproducts in $\coh_{stable}$ is difficult
 according to \cite{Gir89}. However not in $\SetSys$ and
 $\SetSys_{lin}$. It is 
 because the morphisms of the two categories can represent 
nondeterministic computations as we saw in the
 proof of Section~\ref{sec:closureop}.

\begin{theorem}\label{thm:catmain}
\begin{enumerate}
\item
$\SetSys$ and $\SetSys_{lin}$ are indeed complete
categories with all finite coproducts. 
In $\RC$ and $\RC_{lin}$, for objects $\inj \L_j$ $(j\in J)$,
 the coproduct is
\begin{equation}\inj \L := \bigoplus_{j\in J}\;\inj \L_j\ \mbox{where}\
 \L=\left\{\ L\times\{j\}\;;\; L \in \L_ j,\  j \in J \,\right\},\label{def:coproduct}
\end{equation} 
and the product is $\inj\left( \widetilde\biguplus_{j\in J}\L_j\right)$.

\item
 A following $\iota$ is a full functor from $\coh_{stable}$  to 
		      $\SetSys_{seq}$ : 
\begin{eqnarray*}
\iota(\A) = \inj \A,\qquad\quad
\iota(\A\stackrel{F}{\to}\B) = \inj
		      \A\stackrel{\inj^{-1}}\to \A \stackrel{F}{\to} \B
		      \stackrel{\inj}{\to} \inj \B.
\end{eqnarray*} 
\end{enumerate} 
\end{theorem}
\begin{proof}See \ref{sec:cat}. \qed\end{proof}

In $\SetSys_{seq}$, the dimension of an object is a categorical notion.
\begin{theorem}\label{thm:categorical} If $\C$ and $\D$ are isomorphic objects in
 $\SetSys_{seq}$,  $\dim \inj^{-1}\C=\dim \inj^{-1} \D$.
\end{theorem} 
\begin{proof} Because one object is the image of the other object by a
 sequential function, the former dimension is less than or equal to the
 latter dimension by Theorem~\ref{cor:akama}~\eqref{n1}.\qed\end{proof}

Following proves a part of Proposition~\ref{prop:a} by
Moriyama-Sato~\cite{735102}. 

\begin{theorem}\label{lem:coproduct} If $\L$ and $\M$ are
 \textsc{fess}s, 
 $\dim \inj^{-1}\left(\inj \L_1 \oplus \inj \L_2\right)=\max(\dim \L_1, \dim \L_2)$ and
 the union $\L_1\cup\L_2$ is an \textsc{fess}.
\end{theorem}

\begin{proof}
Because
 $\inj^{-1}\left(\inj \L_1 \oplus \inj \L_2\right)=\{ L\times\{j\}\;;\; L\in
 \L_j, j=1,2\}$,  any production sequence of it is exactly written as 
 $\bigl \langle \langle (t_0, j),
 N_1\times \{j\}\rangle, \langle (t_1, j), N_2\times\{j\}\rangle,
 \ldots, \langle (t_{n-1}, j),
 N_n\times \{j\}\rangle \bigr\rangle$ for some $n$,  $N_i\in \L_j$
 $(1\le i\le n)$ and $t_i\in N_i$ $(0\le i\le n-1)$.
 Therefore, $\Prd{\inj^{-1}\left(\inj \L_1 \oplus \inj \L_2\right)}$ is the
 disjoint sum of $\Prd{\L_1}$ and  $\Prd{\L_2}$, from which the
 conclusion follows.

The second assertion is because $\inj(\L\cup\M)$ is the direct image by the
 monotone, continuous function $\op_{R_2}:\inj\L\oplus\inj\M\to \inj\left(\L\cup\M\right)$ of $\inj \L \oplus \inj \M$ where $R_2$ is defined in
 \eqref{prod_to_ewunion}.
\qed
\end{proof}

\begin{definition}
\begin{enumerate}
\item
Define a contravariant functor $\Ss$ from $\QO_{FinSim}$
 to $\SetSys_{lin}$ as follows. Let $\X=(X,\preceq), \Y=(Y,\sqsubseteq)$ be objects of $\QO_{FinSim}$. Put
 $\ss{\X}:=\inj ( \closure{\X})$. For each morphism $R$ from
 $\X$ to $\Y$, let
 $\ss{R}$ be the monotone, linear, continuous function $\op_{\hat{R}}:\ss{\Y}\to
\ss{\X}$ with the trace $\hat{R}=\{(x,\{y\})\;;\; R(x,y)\}$.

\item Define a contravariant functor $\mathrm{Qo}$ from
 $\SetSys_{lin}$ to $\QO_{FinSim}$ as follows. Let $\C$ be an object of
 $\RC_{lin}$. Put $\Qo{\C}:=\qo{\inj^{-1}\C}$. For each morphism
 $\op_R:\D\to \C$ in $\SetSys_{lin}$, let $\Qo{\op_R}$ be a finitely
     branching simulation
 $\check{R}:=\{ (x, y)\;;\; R(x,\,\{y\})\}\subseteq \fld{\C}\times\fld{\D}$ of $\QO_{FinSim}$.
\end{enumerate} 
\end{definition} 

\begin{lemma}\label{lem:functor}
\begin{enumerate}
\item $\Ss$  is indeed a functor       $\QO_{FinSim}^{\mathrm{op}}$ from to
$\SetSys_{lin}$ .
\item
$\mathrm{Qo}$ is indeed a functor from $\SetSys_{lin}$ to
      $\QO_{FinSim}^{\mathrm{op}}$.
\end{enumerate} 
\end{lemma}
\begin{proof}
(1) For every morphism $R:\X\to\Y$ of $\QO_{FinSim}$,
$\ss{R}\left[ \inj ( \closure{\Y})\right] =\op_R\left[\inj \left( \closure{\Y}\right)\right]$
 is $\inj\left(\widetilde{R^{-1}}\left[\closure{\Y}\right] \right)$ by
 \eqref{special},
 a subset of $\inj\left(\closure{\X}\right)$ by 
Lemma~\ref{lem:simulation} with $R$ being a simulation. Thus $\ss{R}$ is
 indeed a function from $\inj\left(\closure{\Y}\right)$ to
 $\inj\left(\closure{\X}\right)$. The functoriality is because 
\begin{eqnarray}
\op_{S \circ R}(g)(x)=&
 \bigvee_{R(x,y)}\bigvee_{S(y,z)} (g(z)=1) =
 \bigvee_{R(x,y)}\left(\op_S(g)(y)=1\right)\nonumber\\
=&\op_R\left(\op_S(g)\right)(x)=
 \left(\op_R\circ\op_S\right)(g)(x) . \label{demo}
\end{eqnarray} 

\bigskip
(2) Firstly, we establish the well-definedness of $\mathrm{Qo}$.
For any finitely branching relations $R,S\subseteq X\times Y$,
 $\op_R=\op_S$ implies $R=S$. To see it,
suppose $\op_R=\op_S$. If $R(x,y)$, then
 $\op_R(1_{\{y\}})(x)=1=\op_S(1_{\{y\}})(x)=\bigvee_{S(x,y')}(y'=y)$,
 which implies $S(x,y)$. Therefore $R=S$. 

Next, $\mathrm{Qo}$ preserves the identity morphism, because
      $\Qo{\op_{\Delta}}=\Delta$ for every set $A$ and for every
      diagonal relation on $A\times
      A$. $\Qo{\op_R\circ\op_S}=\Qo{\op_S}\circ\Qo{\op_R}$ follows from \eqref{demo}.\qed
\end{proof} 

%

According to \cite[Theorem~IV.1.1, Theorem~IV.1.2]{MR1712872},
a functor $G:A\to X$ is a left adjoint functor to a functor $F:X\to A$
if and only if there are natural transformations
$\eta:\mathrm{Id}_X\stackrel{\cdot}{\to} GF$ and
$\varepsilon:F G\stackrel{\cdot}{\to}\mathrm{Id}_A$ such that both the
following composites are the identity natural transformations (of $G$,
resp. $F$.) 
\begin{equation}G\stackrel{\eta G}{\longrightarrow} G F G\stackrel{G\varepsilon}{\longrightarrow} G,
\quad
F\stackrel{F\eta }{\longrightarrow} F G F \stackrel{\varepsilon
F}{\longrightarrow} F. \label{units}
\end{equation}
$\eta$ is called the \emph{unit} and $\varepsilon$ is called the
\emph{counit}.
The \emph{opposite category} of a category $A$ is denoted by $A^{\mathrm{op}}$.

\begin{theorem}\label{thm:adjunction}  The functor $\mathrm{Qo}:\SetSys_{lin}\to
 \QO_{FinSim}^{\mathrm{op}}$ is a left adjoint functor to the functor
 $\Ss:\QO_{FinSim}^{\mathrm{op}}\to \SetSys_{lin}$ where the counit of the
 adjunction is the identity natural transformation of the identity functor $\mathrm{Id}_{\QO_{lin}^\mathrm{op}}$.
\end{theorem}
\begin{proof}By Theorem~\ref{thm:repre}~\eqref{assert:otp} and the
 definition, the composite $\mathrm{Qo}\circ\Ss$ is the identity
 functor of $\QO_{lin}$ .  Define the unit $\eta_\C:\C\to \ss{\Qo{\C}}$
 by the inclusion map. Then \eqref{units} follows immediately.\qed
\end{proof}

\begin{corollary}\label{cor:last}The functor
$\Ss$ is an injective-on-objects, full and faithful functor from
 $\QO_{FinSim}$  to $\SetSys_{lin}^{\mathrm{op}}$. Moreover $\otp (\X)  = \dim \inj^{-1}\ss{\X}$ for every object $\X$ of $\QO_{FinSim}$. 
\end{corollary}

\begin{proof}The first assertion follows from
 Theorem~\ref{thm:adjunction}, because every right adjoint functor is
 full and faithful whenever every component of the counit is an isomorphism~\cite[Theorem~IV.3.1]{MR1712872}.
The other assertion follows from Theorem~\ref{thm:repre}~\eqref{assert:otp}.
\qed\end{proof}

\section*{Acknowledgement} The author sincerely thanks Dr.~Matthew de
 Brecht, Hajime Ishihara, and
 anonymous referees.
This work is partially supported by Grant-in-Aid for Scientific
Research (C) (21540105) of the Ministry of Education, Culture, Sports,
Science and Technology (MEXT.)

\appendix

%
%
%
%

\section{The categories of set systems and
linear/sequential  monotone, continuous functions}\label{sec:cat}

First we will prove Theorem~\ref{thm:catmain}. $\inj \L$ and $\inj L_j$
are as in the Theorem. See Figure~\ref{diagram}~(right).

 For each  $j \in  J$,
the
 injection $\iota_j:\inj \L_j \to \inj \L $ is $\op_{T_j}$ where 
\begin{displaymath}
T_j := \left\{\left(
 \l x,\,  j \r, \{x\}\right)\;;\; x\in \bigcup \L_j,\ j\in J\,\right\}\  \subseteq\
 \fld{\inj \L} \times\left[\fld{\inj \L_ j }\right]^{<2}. 
\end{displaymath} 
For any set $\left\{\op_{S_j}:\inj \L_j\to \D\ ;\ j\in J\right\}$ of
 morphisms of $\SetSys$, define
\begin{displaymath}
 S\, :=\, \left\{(y, v_j\times\{j\})\;;\;  S_j(y, v_j),\ j\in
 J\;\right\} \ \subseteq\  \fld{\D}\times P\left(\fld{\inj \L}\right),
\end{displaymath}
and a possibly non-continuous function $\discop:\inj \L\to\D$ by 
\begin{equation}
 \discop(h)(y):= \bigvee_{S(y,v)} \bigwedge_{x\in v} h(x)=1.\quad(h\in\inj\L,\ y\in\fld{\D} .) \label{S}
\end{equation} 

\begin{proof}[Theorem~\ref{thm:catmain}](1) The two categories are closed under composition because of 
 \eqref{disjtt}.

The terminal object is $\{\emptyset\}=\P{\emptyset}$. Any monotone,
 continuous function $\op_R:\C\to \P{\emptyset}$ has $R\subseteq
 \emptyset\times \inhfin{\fld{\C}}$ and thus $R=\emptyset$. Actually,
 for any $g\in \C$ and $x\in \emptyset$, we have
 $\op_R(g)(x)=\bigvee_{R(x,v)}\bigwedge_{y\in v} (g(y)=1)$. 

For arbitrary nonempty set $\Lambda$, the product of objects $\inj (\L_\lambda)$
$(\lambda\in\Lambda)$ is just the 
$\C = \inj
\left(\ewuplus_{\lambda\in\Lambda} \L_\lambda\right)$.
For each
$\lambda\in \Lambda$, the projection $\Pi_\lambda:\C\to \inj \L_\lambda$
is $\Pi_\lambda(h):= h\left( \langle \bullet, \lambda\rangle \right)$
for all $h\in \C$. For any $\op_{R_\lambda}:\D\to \inj \L_\lambda$
$(\lambda\in\Lambda)$, the mediating morphism $\op_R$ of
Figure~\ref{diagram} is defined by $R\subseteq \fld{\C}\times
\inhfin{\fld{\D}}$ where $R:=\left\{\left( \l{s},\,{\lambda}\r,\
 v\right)\;;\; 
R_\lambda(s,v)\right\}$.
\begin{figure}
\begin{displaymath}
\begin{xy}
*!C\xybox{
\xymatrix{
\D \ar[dr]^{\op_{R_\lambda}}\ar@{-->}[d]_{\op_R}  &       &\inj \D
 \ar@{-->}[d]_{\widetilde\op} \ar[dr]^{\op'}& & & \D  & \\ 
\ar[r]_{\ \ \Pi_\lambda}
              \inj(\ewuplus_{\lambda\in\Lambda}  \L_\lambda)
 &\inj \L_\lambda &   \inj \N
 \ar[r]_{\op} &\ \inj \M\   \ar@<0.5ex>[r]^{\op_1}\ar@<-0.5ex>[r]_{\op_2}
 & \ \inj \L&\ar@{-->}[u]^{\discop}  
              \inj \L
 &\ar[l]^{\ \ \ \ \ \iota_j=\op_{T_j}}  \ar[ul]_{\op_{S_j}}\inj \L_j \\
h\ar@{|->}[r]&h\left( \langle \bullet, \lambda\rangle \right)&&&&
}}
\end{xy}
\end{displaymath}
\caption{Product of $\inj \L_\lambda$'s is $\inj \left( \ewuplus_{\lambda\in \Lambda} \L_\lambda
 \right)$~(left), the
 equalizer $\inj \N$ is constructed in a standard manner~(middle), and coproduct of
 $\inj \L_j$'s is $\inj \left\{\ \L_j \times\{j\} \ ;\  j \in J
 \,\right\}$ where $j$ ranges over a finite set $J$~(right).
\label{diagram}
}
\end{figure}
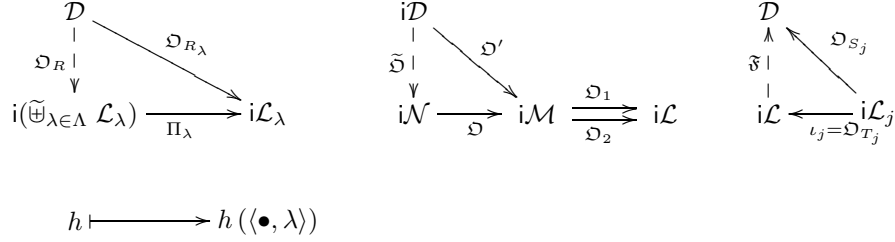
The $\op_R$ is a morphism of $\SetSys$ (and $\SetSys_{lin}$ resp.) if
 $\op_{R_\lambda}$'s are.

The equalizer $\op:\inj \N\to \inj \M$ of a pair of functions $\op_1,\op_2:\inj
 \M\rightrightarrows \inj \L$ is  defined by
\begin{displaymath}
 \inj \N :=\{ g\in \inj \M\;;\; \op_1(g)=\op_2(g)\},\quad \op(g)(x):=
 g(x)\quad(g\in \inj \N, x\in \bigcup\M.)
\end{displaymath}
For Figure~\ref{diagram}~(middle), the mediating morphism
 $\widetilde{\op}:\D\to \inj\N$ is
 defined by $\widetilde{\op}(g)(y)=\op'(g)(y)$ for any $g\in \D$ and
 any $y\in \bigcup \N$.
The initial object is $\emptyset$. Any function from $\emptyset$ to $\C$ is
 the function
 $\emptyset$, which is monotone, continuous because for any $g\in
 \emptyset$ and any $x\in \fld{\C}$, we have $\emptyset(g)(x)= \op_\emptyset(g)(x)=
 \bigvee_{\emptyset(x,v)}\bigwedge _{y\in v} \left( g(y)=1\right)$. 

The existence of a binary coproduct is because the finiteness of $J$
 implies the $\discop$ is indeed a
 morphism of $\SetSys$ ($\SetSys_{lin}$ resp.) if $\op_{S_j}$'s are.

\medskip
(2) According to \cite[Section~8.5]{Gir89}, the stable function $F:\A\to\B$ is
 exactly a
 function from $\A$ to $\B$ having a \emph{trace}. Here the trace of $F$ is the set $R$ of pairs $(x,v)\in
 \bigcup\B\times\inhfin{\bigcup\A}$ such that $v$ is a minimal (and
 actually the minimum) among $L$'s such that $x\in F(L)$. A stable
 function $F$ is recovered from the trace $R$ by
$F(L)=\{x\in \bigcup \B\;;\; \exists v\subseteq L.\ R(x,v)\}$ for all $L
 \in \A$. So $\iota(F)=\inj\circ F\circ\inj^{-1}$ is written as $\op_R$. Because $v$ is minimum, and
 is in particular unique, $\op_R$ is sequential, i.e., $\op_R\in
 \SetSys_{seq}$. So $\iota$ is indeed  well-defined. We can easily see
 that $\iota$ is indeed a functor.

Next we verify that the functor $\iota$ is indeed full.
Let  $\op_R:\iota(\A)\to \iota(\B)$ be a morphism of the category $\SetSys_{seq}$. Recall $\op_R(1_L)(x)=1\iff \exists
 v\subseteq L.\ R(x,v)$. Because $\op_R$ is sequential,
 each $x$ has at most one $v$ such
 that $R(x,v)$.  So  $R$ is the set of pairs
 $(x,v)$ such that $v$ is minimum among $L$'s such that
 $\op_R(1_L)(x)=1$. Thus $\inj^{-1} \circ F\circ \inj$ is the stable
 function with the trace being $R$.
\qed
\end{proof}


\begin{lemma}\label{lem:infcopro}
None of $\SetSys$, $\SetSys_{lin}$ and $\SetSys_{seq}$ does not have 
the object $\inj \L$ of
 \eqref{def:coproduct} as a coproduct if $J$ is infinite. Even
 $\SetSys_{seq}$ does not for $2\le \#J\le\infty$.

\end{lemma}
\begin{proof} We show that $\discop:\inj \L\to\D$ of \eqref{S} is not a
 morphism of $\SetSys$, when
\begin{equation}
S_\mu:=\{ (y,\emptyset)\;;\;y\in\fld{\D}\}\subseteq
 \fld{\D}\times\left[\fld{\inj \L}\right]^{<2},\quad (\mu\in J) \label{akm1}
\end{equation}
 Let $y\in \fld{\D}$.
Because $J$ is infinite but $S\subseteq\fld{\D}\times\inhfin{\fld{\inj \L}}$, there is $\mu\in J\setminus \{ j\;;\; \exists v\in \inhfin{\fld{\inj \L}}.\  S(y,v)\,
 \wedge \exists \xi\in v\exists a.\ \xi=\pair{a}{j}\ \}$. Let $f\in \inj\L_\mu $ such that $f(x)=1$ for all $x\in \fld{\inj
 \L_\mu}$. Then by Figure~\ref{diagram} (right), we have
$\left(\op_S\circ \iota_\mu\right) (f)(y) =  \op_{S_\mu} (f)(y)$. Therefore
\begin{displaymath}
 \bigvee_{S(y,v)}\bigwedge_{\xi\in v}\ (\iota_\mu (f)(\xi) = 1)  =  \bigvee_{S_\mu(y,u)}\bigwedge_{x\in u}\ (f(x) = 1) .
\end{displaymath}
Here $\xi\in v$ is written as $\xi=\langle a, j\rangle$ for some
 $j\ne\mu$. So $\iota_\mu (f)(\langle a, j\rangle) =0$, which implies
 the left-hand side is 0. But, the
 right-hand side is 1 by \eqref{akm1}.
\qed
\end{proof} 

It is difficult to relate $\dim (\L \ewprod \M)$ with $\dim (\L \ewunion
\M)$. When $\L$ and
$\M$ are both coherence spaces, 
$\L\ewprod\M$ is the ``tensor product'' $\L\otimes\M$. 

\begin{lemma}\label{lem:no_ewprod_ewunion}Let $\L$ and $\M$ be set
 systems with $\bigcup \L$ infinite and $\bigcup\M\ne\emptyset$.   Then,
\begin{enumerate}
\item \label{assert:mult_to_add}
There is no monotone, continuous function $\op : \inj (\L\ewprod \M) \to
\inj( \L\ewunion \M )$ such that $\op(1_{L\times M})= 1_{L\cup M}$ for all
 $L\in \L$ and $M\in \M$.

\item\label{assert:mult_is_not_product}
There is no monotone, continuous function $\op:\inj (\L \ewprod \M) \to
\inj \M$ such that $\op(1_{L\times M})= 1_M$ for all
 $L\in \L$ and $M\in \M$.
\end{enumerate}
\end{lemma}
\begin{proof}\eqref{assert:mult_to_add} Let $X:=\bigcup\L$ and
 $Y:=\bigcup\M$. 
 Assume there is such $\op$. Then for each $s\in X\cup Y$, there
 exists a positive Boolean formula $B_s$ over $\{v_{(x,y)}\;;\;
 x\in X,\ y\in Y \}$, such that $\op(1_{L\times M})(s)$ is the truth value of $B_s$ under
 the truth assignment $1_{L \times M}$ for all $L\in\L$ and all
 $M\in \M$.
 Choose some $y\in Y$. 
There is a variable $v_{(x,y)}$ such that it does not appear
 $B_y$, because $X$ is infinite. Therefore the truth value of $B_y$ under
 the truth assignment $1_{\{(x,y)\}}$ is 0 because $B_y$ does not
 contain negations of Boolean variables. On the other hand
 $\op(1_{\{(x,y)\}})(y) = 1_{\{x\}\cup\{ y\}}(y) = 1$. Contradiction. 
The assertion \eqref{assert:mult_is_not_product} is similarly proved.
\qed
\end{proof}

The bang operator of a coherence space have following counterparts in $\SetSys$:
\begin{displaymath}
 !\inj \L := \inj !\L 
\end{displaymath}
where the `!' in the right-hand side is defined in Theorem~\ref{thm:bang}.
Then
$!  \L\ \ewprod\ ! \M$ is isomorphic to $ \L \ewuplus \M$, as in the case of $\coh_{stable}$.
The duality operator of a coherence space, however, seems to have no exact counterpart in
$\SetSys$, when we take an \textsc{fe} seriously. Since the elementwise complement of an \textsc{fess} is not
necessarily an \textsc{fess},  the complement operation seems useless in
defining the duality operator in $\SetSys$. So let us examine the exchange of Teacher and Learner. To be precise,
For a set system $\L$ and $x\in \L$,  put $\L(x):=\left\{
 L\in \L\;;\; x\in L\right\}$, and $\L^\perp:=\left\{ \ \L(x) ;\ x \in
 \bigcup \L\ \right\}$. Then
$\bigcup\left(\L^\perp  \right):= \L\setminus\{\emptyset\}$.
If $\L$ is the class of open
 sets of a sober space, then  $\left(\L^\perp\right)^\perp$ is isomorphic
 to $\L$ in $\SetSys$. Since $L\in \L^\perp(x)$ iff $x\in L$,
\begin{eqnarray*}
&& \Langle \pair{t_0}{L_1}, \pair{t_1}{L_2},\ldots,
 \pair{t_{l-2}}{L_{l-1}},\pair{t_{l-1}}{L_l}\R \in\Prd{\L} \\
&&\iff\\
&& \Langle \pair{L_l}{\L(t_{l-1})}, \pair{L_{l-1}}{\L(t_{l-2})},\ldots,
\pair{L_2}{\L(t_1)},\pair{L_1}{\L(t_0)}\R \in\Prd{\L^\perp}
\end{eqnarray*}
We have an embedding from $\Prd{\L}$ to $\Prd{\left(\L^\perp\right)^\perp}$,
by $\bigcap\bigcap\left(\left(\L^\perp\right)^\perp(L)\right)=L$.
 Thus $\dim \L\le \dim \left(
\L^\perp\right)^\perp$.

Further categorical structures will be studied elsewhere.

\end{document}